\documentclass{amsart}

\usepackage[english]{babel}
\usepackage[utf8]{inputenc}
\usepackage[T1]{fontenc}

\usepackage{amsmath}
\usepackage{amssymb}
\usepackage{mathtools}
\usepackage{amsthm}
\usepackage{physics}
\usepackage{cite}

\usepackage{booktabs}
\usepackage{rotating}
\usepackage{tabularx}
\usepackage{multirow}

\usepackage{siunitx}

\usepackage{graphicx}
\usepackage{subfig}
\usepackage{wrapfig}
\usepackage{xcolor}
\usepackage{subfig}
\usepackage{tikz}

\usepackage{etoolbox} 
\AtBeginEnvironment{pmatrix}{\everymath{\displaystyle}}
\AtBeginEnvironment{bmatrix}{\everymath{\displaystyle}}

\DeclareMathOperator{\gdc}{gdc}
\DeclareMathOperator{\Id}{Id}
\DeclareMathOperator{\ud}{d}

\newcommand{\Z}{\mathbb{Z}}
\newcommand{\N}{\mathbb{N}}

\newcommand{\dP}{\mathrm{d}P}
\newcommand{\lequiv}{\stackrel{\text{t.d.}}{\equiv}}

\renewcommand{\vec}[1]{\mathbf{#1}}
\renewcommand{\tilde}{\widetilde}

\newcommand\lagrangeprime[1]{^{%
\ifcase#1 \or\prime\or\prime\prime\or\prime\prime\prime\else\mathrm{\romannumeral #1}\fi}}

\numberwithin{equation}{section}

\newcounter{eqlist}

\newcounter{eqlistred}

\allowdisplaybreaks
\makeatletter 
\renewcommand{\pdv}[2]{\begingroup 
\@tempswafalse\toks@={}\count@=\z@ 
\@for\next:=#2\do 
{\expandafter\check@var\next\@nil
 \advance\count@\der@exp 
 \if@tempswa 
   \toks@=\expandafter{\the\toks@\,}%
 \else 
   \@tempswatrue 
 \fi 
 \toks@=\expandafter{\the\expandafter\toks@\expandafter\partial\der@var}}%
\frac{\partial\ifnum\count@=\@ne\else^{\number\count@}\fi#1}{\the\toks@}%
\endgroup} 
\def\check@var{\@ifstar{\mult@var}{\one@var}} 
\def\mult@var#1#2\@nil{\def\der@var{#2^{#1}}\def\der@exp{#1}} 
\def\one@var#1\@nil{\def\der@var{#1}\chardef\der@exp\@ne} 
\makeatother

\theoremstyle{plain}
\newtheorem{theorem}{Theorem}[section]
\newtheorem{proposition}[theorem]{Proposition}
\newtheorem{corollary}[theorem]{Corollary}
\newtheorem{lemma}[theorem]{Lemma}
\theoremstyle{definition}
\newtheorem{definition}{Definition}[section]
\theoremstyle{remark}
\newtheorem{remark}{Remark}[section]
\theoremstyle{remark}

\theoremstyle{definition}
\newtheorem*{problem*}{Problem}

\begin{document}
\title{On the inverse problem of the discrete calculus of variations}

\author{Giorgio Gubbiotti}
\address{School  of  Mathematics  and  Statistics  F07,  The  University  of  Sydney,  NSW  2006, Australia}
\email{giorgio.gubbiotti@sydney.edu.au}

\date{\today}

\begin{abstract}
    In this paper we present an algorithm to find
    the discrete Lagrangian for an autonomous recurrence relation of
    arbitrary even order $2k$ with $k>1$.
    The method is based on the existence of a set of differential
    operators called annihilation operators which can be used to
    convert a functional equation into a system of linear partial
    differential equations.
    This completely solves the inverse problem of the calculus 
    of variations in this setting.
\end{abstract}

\maketitle

\section{Introduction}

One of the most powerful tools in Mathematical Physics
since the times of Euler and Lagrange is the  \emph{calculus of variations}.
The variational formulation of mechanics where the equations
of motion arise as the minimum of an \emph{action functional}
(the so-called \emph{Hamilton's principle}), is fundamental in
the developing of theoretical mechanics and its foundations are
present in each textbook on this subject \cite{Gold2002,LandauMech,Whittaker}.
Beside this, the application of calculus of variations goes beyond
mechanics as many important mathematical problems, e.g. the isoperimetrical
problem and the catenary, can be formulated in terms of calculus of
variations.
For a more complete outlook on the calculus of variations, its scopes
and its applications we refer to standard textbooks on the
subject \cite{Dacorogna2015book,Forsyth1960book,GelfandFomin1963book,JostLiJost1998book,Sagan1992book}.

An important problem regarding the calculus of variations
is to determine when a system of differential equations are the 
Euler--Lagrange equations for some variational problem.
This problem is called \emph{the inverse problem of the calculus of variations}
and has a long and interesting history.
It was was first addressed by Jacobi in the case of second-order 
scalar ordinary differential equations
\cite{Jacobi44a,Jacobi45,Jacobi1846,JacobiVD}. 
In this case it turns out that the answer is that such kind
of equations admit infinitely many inequivalent Lagrangians.
Jacobi's proof use the so-called Jacobi Last Multiplier, which can be explicitly 
used to construct Lagrangians.
This proof can be found in \cite{Whittaker}.
The same result was found also by different authors after Jacobi, 
usually with no explicit mention of his work,
see \cite{Darboux1915,Helmholtz1887}.
The general case of this problem remains unsolved, 
whereas several important results for particular cases where
presented during the XX century.
In particular we mention that, under certain restrictions,
the general case can be addressed by a method developed by the
Italian mathematician V. Volterra \cite{Volterra1916}.
Such method can be understood in a modern formalism through
the theory of variational complexes and homotopy \cite{Vainberg1964}.
Volterra's result is therefore now known as \emph{the homotopy formula}.
A complete solution in the case of the systems of two 
second-order ordinary differential equations was given by
Douglas in 1941 \cite{Douglas1941}.
It is worth to note that Douglas' solution have been interpreted
in terms of geometrical calculus in \cite{Crampinetal1994,SarletThompsonPrince2002}
using the important advances on the subject given in \cite{AndersonThompson1992}.
In the scalar case necessary and sufficient conditions on 
the existence of a Lagrangian are known up to eight-order equations.
Indeed in \cite{Fels1996} such conditions were derived in the case 
of fourth-order equations, whereas in \cite{Juras2001} a set of conditions
for the existence of a Lagrangian for sixth- and eight-order equations
was given.
For systems of $n$ second-order differential equations
some results on \emph{radially symmetric systems} exist 
\cite{HenneauxShepley1982,Caviglia1985}.
For a complete, yet accessible historical introduction on the subject we refer 
to \cite[Chap. 4]{Olver1986}.
For a more updated exposition on the results obtained on the subject
in the last 30 years we refer to the reviews \cite{Saunders2010,Vitolo2007book}
and references therein.

We mention that the relation of the Jacobi Last Multiplier 
with the inverse problem of the calculus of variations was
often neglected in literature. The Jacobi Last Multiplier itself 
was rediscovered several times by researchers unaware of its properties.
For an historical perspective on this subject we refer to 
\cite{jlm05,nuctam_2lag,laggal} and references therein.
Moreover, we observe that the Jacobi Last Multiplier can be
used to find Lagrangians for differential equations of
order higher than two \cite{NucciArthurs2010} and for
non-dissipative systems of second-order equations 
\cite{Rao1940,CP07Rao1JMP}.

In this paper we give some condition 
on the existence of a Lagrangian in the discrete scalar setting.
To be more precise, we will reduce the condition of existence
of a Lagrangian for a scalar difference equation of arbitrary
even order $2k$, with $k>1$ to the solution of an overdetermined 
system of linear partial differential equations.
The fact that the system is linear and overdetermined, means that
in general it possible to solve it without an excessive effort.
Our solution to the problem is constructive:
solving such system one can either find the Lagrangian or conclude
that it does not exist.
This result pave the way to various applications, for instance in
the classification of difference equations based on their variational
structure, or their integrability properties.
This approach has been already applied to understand how integrability
arises in a family of fourth-order maps with two given invariants
\cite{GJTV_class,GJTV_sanya}.

The inverse problem of discrete calculus of variations 
was considered from the point of view of variational
complexes and homotopy in \cite{HydonMansfield2004}.
In such paper a result analogous to the homotopy formula for the 
continuous case was proved.
We mention that an early description of the difference variational
complex appeared in \cite{Kupershmidt1985book}.

The plan of the paper is following: in Section \ref{sec:dlagr} we introduce
the basis of the discrete calculus of variations for scalar difference
equations. 
We give a self-contained account of all the fundamental results 
we are going to use based on the original papers on the subject
\cite{Bruschietal1991,Logan1973,TranvanderKampQusipel2016,Veselov1991}.
In Section \ref{sec:method} we present our main result which is
given by Theorem \ref{thm:elexist}.
This theorem, along with its Corollary \ref{cor:nolagr}, allows
to prove or disprove the existence of a Lagrangian in our setting.
Section \ref{sec:examples} is then devoted to examples.
We present several examples from the recent literature, 
even of equations of arbitrarily high order.
Finally in Section \ref{sec:concl} we give some conclusions
and outlook.
In particular we discuss the possibility of using the method presented 
in section \ref{sec:method} in the framework of the 
\emph{geometric integration theory} 
\cite{BuddIserles1999,BuddPiggot2003,KrantzParks2008},
and we present a comparison with the homotopy approach given in
\cite{HydonMansfield2004}.

\section{Discrete Lagrangians and their basic properties}
\label{sec:dlagr}

In this section we introduce discrete Lagrangians and those of 
their properties we are going to use throughout the paper.
Some of these properties were discussed already in \cite{Logan1973}.
Later accounts on these properties of discrete Lagrangians
can be found in \cite{Veselov1991,Bruschietal1991} and more recently
in \cite{TranvanderKampQusipel2016}.
We refer to these papers and references therein for a more
complete overview on the subject.

Throughout this paper we are going to consider consider 
recurrence relation of even order $2k$, equivalently
called a scalar difference equation, i.e. functional equations
of the following from:
\begin{equation}
    x_{n+k} = F\left( x_{n+k-1},x_{n+k-2},\dots,x_{n-k} \right).
    \label{eq:recgen}
\end{equation}
where $k\geq1$.
To be a well-posed equation of order $2k$ we impose the 
\emph{non-degeneracy condition}:
\begin{equation}
    \pdv{F}{x_{n-k}}\neq 0.
    \label{eq:nondeg}
\end{equation}
Throughout the paper we will \emph{always} consider this condition to be
satisfied.

A \emph{discrete action} of order $k$ is a linear functional of the form:
\begin{equation}
    S\left[ x_{n} \right] = \sum_{n\in\Z} L\left( x_{n+k},x_{n+k-1},\dots,x_{n} \right).
    \label{eq:daction}
\end{equation}
The summand function 
\begin{equation}
    L = L\left( x_{n+k},x_{n+k-1},\dots,x_{n} \right)
    \label{eq:lagrgen}
\end{equation}
is called a \emph{discrete Lagrangian}.

Usually we think difference equation \eqref{eq:recgen}
and discrete Lagrangians \eqref{eq:lagrgen} to be autonomous,
i.e. not explicitly dependent on $n$.
For this reason in the general fomul\ae\ \eqref{eq:recgen} 
and \eqref{eq:lagrgen} we omitted the index $n$ in the hand side.
However, we underline that all the reasoning presented 
in this paper also work in the non-autonomous case, 
with the appropriate care. 
In particular we will underline with appropriate remarks when a result 
can be simplified in the autonomous setting.
On the other side, when a result requires to be be discussed in 
the non-autonomous setting we will always place the subscript $n$
and explain why it is needed.

From the variation of the discrete action \eqref{eq:daction}
we obtain that the stationary points of the discrete action 
satisfy the following difference equation \cite{Logan1973}:
\begin{equation}
    \sum_{l=0}^{k} T^{-l}_{n}\pdv{L}{x_{n+l}}\left( x_{n+k},x_{n+k-1},\dots,x_{n} \right) = 0.
    \label{eq:elgen}
\end{equation}
In formula \eqref{eq:elgen} $T_{n}$ is the translation operator
acting on any function $f_{n}=f_{n}\left( x_{n+l},\dots,x_{n+m} \right)$
depending on a finite number of shifts of $x_{n}$ and possibly on
the independent variable $n$ as:
\begin{equation}
    T_{n} f_{n}\left( x_{n+l},\dots,x_{n+m} \right) = 
    f_{n+1}\left( x_{n+l+1},\dots,x_{n+m+1} \right).
    \label{eq:traslop}
\end{equation}

Equation \eqref{eq:elgen} is called the discrete Euler--Lagrange equation 
corresponding to the discrete Lagrangian $L$ \eqref{eq:lagrgen}.
When no ambiguity is possible, we will simply address equation
\eqref{eq:elgen} as the discrete Euler--Lagrange equation.
The left hand side of the discrete Euler--Lagrange equations \eqref{eq:elgen}
is sometimes called the \emph{variational derivative} of the
action \eqref{eq:daction} and denoted by $\delta S/\delta x_{n}$.

We say that a discrete Lagrangian \eqref{eq:lagrgen} is a 
discrete Lagrangian  for the difference equation \eqref{eq:recgen} 
if its discrete Euler--Lagrange equations \eqref{eq:elgen} coincide 
with \eqref{eq:recgen}.

\begin{remark}
    The expression of the Euler--Lagrange equation given in
    equation \eqref{eq:elgen} is valid both for autonomous and
    non-autonomous Lagrangians.
    In particular in the case of autonomous Lagrangians 
    $L=L\left( x_{n+k},\dots,x_{n} \right)$ its expression can be simplified
    to:
    \begin{equation}
        \sum_{l=0}^{k} \pdv{L}{x_{n}}\left( x_{n+k-l},x_{n+k-1},\dots,x_{n-l} \right) = 0.
        \label{eq:elgenaut}
    \end{equation}
    On the other hand if the Lagrangian depends explicitly on $n$,
    $L=L_{n}$ we can write formula \eqref{eq:elgen} as:
    \begin{equation}
        \sum_{l=0}^{k} \pdv{L_{n-l}}{x_{n}}\left( x_{n+k-l},x_{n+k-1},\dots,x_{n-l} \right) = 0.
        \label{eq:elgennonaut}
    \end{equation}
    \label{eq:elautnonaut}
\end{remark}

When developing the theory of discrete Lagrangians it is good to consider 
both expressions \eqref{eq:elgen} and \eqref{eq:elgennonaut} as the former
is more abstract and useful for theoretical purposes, while the latter is
more useful in explicit computations.

In particular formula \eqref{eq:elgen} allow us to think of the variational 
derivative as a linear operator acting on functions of the form \eqref{eq:lagrgen}.
That is, we can state the following definition:
\begin{definition}
    The linear differential-difference operator:
    \begin{equation}
        \mathcal{E} = 
        \sum_{l=0}^{k} T^{-l}_{n}\pdv{}{x_{n+l}}.
        \label{eq:EulerOperator}
    \end{equation}
    is called the \emph{Euler operator}.
    \label{def:euler}
\end{definition}

The Euler--Lagrange equation is written in terms of the Euler operator
\eqref{eq:EulerOperator} as $\mathcal{E}(L)=0$.
The following important result about the kernel of the Euler operator 
\eqref{eq:EulerOperator} holds:
\begin{lemma}
    The kernel of the Euler operator \eqref{eq:EulerOperator} is the
    space of \emph{total differences}, i.e. of functions 
    $g_{n} = g_{n}\left(x_{n+k},\dots,x_{n}\right)$ such that there exists
    a function $f_{n}=f_{n}\left(x_{n+k-1},\dots,x_{n}\right)$ such that:
    \begin{equation}
        g_{n}\left(x_{n+k},\dots,x_{n}\right) = 
        (T_{n}-\Id)f_{n}\left( x_{n+k-1},\dots,x_{n} \right),
        \label{eq:totdiff}
    \end{equation}
    that is, the image of the operator $T_{n}-\Id$.
    Symbolically we write:
    \begin{equation}
        \ker \mathcal{E} = \Im\left( T_{n}-\Id \right).
        \label{eq:kerE}
    \end{equation}
    \label{lem:zeroder}
\end{lemma}

A proof of Lemma \ref{lem:zeroder} in the context of difference variational
complex can be found in \cite{HydonMansfield2004,Kupershmidt1985book}.
Here we present a simple proof which does not require any knowledge
of the geometrical theory of discrete calculus of variations.
However, is worth showing such proof since 
the computational techniques employed
will be used when constructing discrete Lagrangians in
section \ref{sec:examples}.

\begin{proof}
    First we prove the inclusion $\Im\left( T_{n}-\Id \right)\subset\ker \mathcal{E}$.
    Assume $g_{n}\in \Im\left( T_{n}-\Id \right)$ and apply $\mathcal{E}$ to it:
    \begin{equation}
        \mathcal{E}\left[ g_{n}\left( x_{n+k},\dots,x_{n} \right) \right]
        =
        \mathcal{E}\left[ \left( T_{n}-\Id \right)f_{n}\left( x_{n+k-1},\dots,x_{n} \right) \right].
        \label{eq:applfn}
    \end{equation}
    After some trivial algebra using the definition of the
    Euler operator \eqref{eq:EulerOperator} we find following identity:
    \begin{equation}
        \mathcal{E}\left[ g_{n}\left( x_{n+k},\dots,x_{n} \right) \right]
        =
        \pdv{}{x_{n}}
         \sum_{l=0}^{k}  
         \left[f_{n+1-l}\left( x_{n+k-l},\dots,x_{n-l+1} \right)
            -f_{n-l}\left( x_{n+k-l-1},\dots,x_{n-l} \right)
        \right].
        \label{eq:applfn2}
    \end{equation}
    The sum on the right hand side of equation \eqref{eq:applfn2} is
    telescopic, therefore its value is given by:
    \begin{equation}
        \mathcal{E}\left[ g_{n}\left( x_{n+k},\dots,x_{n} \right) \right]
        =
        \pdv{}{x_{n}} 
         \left[f_{n+1}\left( x_{n+k},\dots,x_{n+1} \right)
            -f_{n-l}\left( x_{n-1},\dots,x_{n-k} \right)
        \right]=0.
        \label{eq:applfn3}
    \end{equation}
    This proves that $g_{n}\in \ker\mathcal{E}$.

    Now, we prove the reverse inclusion $\ker \mathcal{E}\subset\Im\left( T_{n}-\Id \right)$.
    We proceed by induction on the number of points.
    \\
    \noindent
    \emph{Case $k=1$.}
    Assume $g_{n}=g_{n}\left( x_{n+1},x_{n} \right)\in\ker \mathcal{E}$, that is:
    \begin{equation}
        \pdv{g_{n}}{x_{n}}\left( x_{n+1},x_{n} \right) 
        + \pdv{g_{n-1}}{x_{n}}\left( x_{n},x_{n-1} \right)=0.
        \label{eq:gndef}
    \end{equation}
    Differentiating the previous equation with respect to $x_{n+1}$ we obtain
    that $g_{n}$ solves the following partial differential equation for all
    $n$:
    \begin{equation}
        \pdv{g_{n}}{x_{n+1},x_{n}}\left( x_{n+1},x_{n} \right)=0.
        \label{eq:gndef2}
    \end{equation}
    This last equation imply:
    \begin{equation}
        g_{n}\left( x_{n+1},x_{n} \right) =
        g_{n}^{(1)}\left( x_{n+1} \right)+g_{n}^{(2)}\left( x_{n} \right).
        \label{eq:gndefsol}
    \end{equation}
    Taking advantage of the arbitrariness of $g_{n}^{(1)}$ and $g_{n}^{(2)}$ we
    can write 
    $g_{n}^{(2)}\left( x_{n+1} \right)=G_{n}\left( x_{n} \right)
    -g_{n}^{\left( 1 \right)}\left( x_{n} \right)$, that is:
    \begin{equation}
        g_{n}\left( x_{n+1},x_{n} \right) =
        G_{n}\left(x_{n}  \right)+
        \left(T_{n}-\Id\right)g_{n}^{(1)}\left( x_{n} \right).
        \label{eq:gndefsol2}
    \end{equation}
    As $\left(T_{n}-\Id\right)g_{n}^{(1)}\left( x_{n} \right)\in\Im\left( T_{n}-\Id \right)$
    substituting \eqref{eq:gndefsol2} into \eqref{eq:gndef} we have
    $G_{n}'\left( x_{n} \right)=0$.
    This implies $G_{n}\left( x_{n} \right)=G_{n}$, or in \eqref{eq:gndefsol2}:
    \begin{equation}
        g_{n}\left( x_{n+1},x_{n} \right) =
        G_{n}+
        \left(T_{n}-\Id\right)g_{n}^{(1)}\left( x_{n} \right)=
        \left(T_{n}-\Id\right)\left[ g_{n}^{(1)}\left( x_{n} \right)+\tilde{G}_{n} \right].
        \label{eq:gndefsol3}
    \end{equation}
    In formula \eqref{eq:gndefsol3} we represented $G_{n}$ as a discrete
    integration $G_{n}=\left( T_{n}-\Id \right)\tilde{G}_{n}$, 
    as we can always do when there is no explicit dependence
    on $x_{n}$ and its shifts.
    From \eqref{eq:gndefsol3} we have that for $k=1$ $g_{n}\in\Im\left( T_{n}-\Id \right)$
    with
    \begin{equation}
        f_{n}\left( x_{n} \right) =g_{n}^{(1)}\left( x_{n} \right)+\tilde{G}_{n}.
        \label{eq:fnk2def}
    \end{equation}
    \\
    \noindent
    \emph{Case $k>1$.}
    Assume $g_{n}=g_{n}\left( x_{n+k},\dots,x_{n} \right)\in\ker \mathcal{E}$ for
    $k>1$, and that the property
    $\ker\mathcal{E} \subset \Im\left( T_{n}-\Id \right)$ holds for all function $h_{n}$
    depending on at most $k-1$ points.
    As $g_{n}\in\ker \mathcal{E}$ we have:
    \begin{equation}
        \sum_{l=0}^{k}\pdv{g_{n-l}}{x_{n}}\left( x_{n+k-l},\dots,x_{n-l} \right)=0.
        \label{eq:gndefk}
    \end{equation}
    If we differentiate \eqref{eq:gndefk} with respect to $x_{n+k}$ we obtain:
    \begin{equation}
        \pdv{g_{n-l}}{x_{n+k},x_{n}}\left( x_{n+k-l},\dots,x_{n-l} \right).
        \label{eq:gndefk2}
    \end{equation}
    Reasoning in analogous way as in the case $k=1$ we have that
    we can write the solution of the previous partial differential
    equation as:
    \begin{equation}
        g_{n}\left( x_{n+k},\dots,x_{n} \right) =
        G_{n}\left( x_{n+k-1},\dots,x_{n} \right)+
        \left(T_{n}-\Id\right)\hat{g}_{n}\left( x_{n+k-1},\dots,x_{n} \right).
        \label{eq:gndefksol}
    \end{equation}
    From the fact that $\Im\left( T_{n}-\Id \right)\subset \ker\mathcal{E}$ inserting
    \eqref{eq:gndefksol} into \eqref{eq:gndefk2} we obtain:
    \begin{equation}
        \sum_{l=0}^{k-1}\pdv{G_{n-l}}{x_{n}}\left( x_{n+k-l},\dots,x_{n-l} \right)=0.
        \label{eq:Gndefk}
    \end{equation}
    That is, $G_{n}\left( x_{n+k-1},\dots,x_{n} \right)\in\ker\mathcal{E}$.
    From induction it follows:
    \begin{equation}
        G_{n}\left( x_{n+k-1},\dots,x_{n} \right) =
        \left(T_{n}-\Id\right)\tilde{g}_{n}\left( x_{n+k-2},\dots,x_{n} \right).
        \label{eq:Gndefksol}
    \end{equation}
    Inserting \eqref{eq:Gndefksol} into \eqref{eq:gndefksol} we have
    that $g_{n}\left( x_{n+k},\dots,x_{n} \right)\in\Im\left( T_{n}-\Id \right)$ with
    the function $f_{n}\left( x_{n+k-1},\dots,x_{n} \right)$ given by:
    \begin{equation}
        f_{n}\left( x_{n+k-1},\dots,x_{n} \right)
        =\tilde{g}_{n}\left( x_{n+k-2},\dots,x_{n} \right)
        +\hat{g}_{n}\left( x_{n+k-1},\dots,x_{n} \right).
        \label{eq:fndefk}
    \end{equation}
    This concludes the proof of the lemma.
\end{proof}

\begin{remark}
    We remark that in Lemma \ref{lem:zeroder} the presence of
    the subscript $n$ is necessary when proving the inclusion
    $\ker \mathcal{E}\subset\Im\left( T_{n}-\Id \right)$ with the proposed technique.
    Indeed, if the functions $g_{n}$ and $f_{n}$ are autonomous we have an 
    obstruction to the proof already in the case $k=1$.
    In this case the solution of equation \eqref{eq:gndef2}
    read as:
    \begin{equation}
        g\left( x_{n+1},x_{n} \right) = G\left( x_{n} \right)
        + \left( T_{n}-\Id \right)g_{n}^{(1)}\left( x_{n} \right).
        \label{eq:gdef1aut}
    \end{equation}
    In \eqref{eq:gndef} this implies $G'\left( x_{n} \right)=0$,
    that is $G\left( x_{n} \right)=K_{1}$, with $K_{1}$ a constant.
    The function:
    \begin{equation}
        g\left( x_{n+1},x_{n} \right) = K_{1}
        + \left( T_{n}-\Id \right)g_{n}^{(1)}\left( x_{n} \right),
        \label{eq:gdef1aut2}
    \end{equation}
    is in total difference form summing back equation \eqref{eq:gdef1aut2}:
    \begin{equation}
        f_{n}\left( x_{n} \right) =g_{n}^{(1)}\left( x_{n} \right) 
        + \left( T_{n}-\Id \right)^{-1}\left( K_{1} \right)+K_{2}
        =g_{n}^{(1)}\left( x_{n} \right) +K_{1}n+ K_{2}.
        \label{eq:fdef1aut}
    \end{equation}
    That is, the function $f_{n}$ will necessary depend explicitly on $n$.
    \label{rem:nonautnec}
\end{remark}

Lemma \ref{lem:zeroder} has the following immediate corollary:
\begin{corollary}
    If two discrete Lagrangians $L_{1}$ and $L_{2}$ differ by a total 
    difference, i.e. there exists a function $f_{n}=f_{n}\left( x_{n+k-1},\dots,x_{n} \right)$
    such that:
    \begin{equation}
        L_{2} = L_{1} + \left( T_{n}-\Id \right)f_{n}\left( x_{n+k-1},\dots,x_{n} \right),
        \label{eq:lagrdiff}
    \end{equation}
    then they define the same discrete Euler--Lagrange equations.
    \label{cor:lagrequiv}
\end{corollary}

\begin{proof}
    The thesis follows immediately applying the Euler operator
    \eqref{eq:EulerOperator} to equation \eqref{eq:lagrdiff} and
    noting that the total difference on the right hand side disappears
    using Lemma \ref{lem:zeroder}.
\end{proof}

Using corollary \eqref{cor:lagrequiv} we define the following relation on
discrete Lagrangians: 

\begin{definition}
    Two discrete Lagrangians $L_{1}$ and $L_{2}$ are called
    equivalent, denoted by $\lequiv$, if they differ for a
    total difference, i.e.:
    \begin{equation}
        L_{1} \lequiv L_{2} \iff 
        L_{1} = L_{2} + \left( T_{n} -\Id\right) f\left( x_{n+k-1},\dots,x_{n} \right).
        \label{eq:lequiv}
    \end{equation}
    \label{def:equivlagr}
\end{definition}

Is it easy to prove the following proposition:
\begin{proposition}
    The relation $\lequiv$ is an equivalence relation.
    That is, it possesses the following properties:
    \label{prop:lequiv}
    \begin{description}
        \item[Reflexivity] $L\lequiv L$.
        \item[Symmetry] If $L_{1}\lequiv L_{2}$ then $L_{2}\lequiv L_{1}$.
        \item[Transitivity] If $L_{1}\lequiv L_{2}$ and $L_{2}\lequiv L_{3}$
            then $L_{1}\lequiv L_{3}$.
    \end{description}
\end{proposition}

The proof of Proposition \ref{prop:lequiv} is trivial and it is therefore
omitted.

\begin{remark}
    Proposition \ref{prop:lequiv} implies that Lagrangians are not functions, 
    but rather \emph{equivalence classes} of functions up to the equivalence relation 
    defined by $\lequiv$ \eqref{eq:lequiv}. 
    A Lagrangian equivalent to a constant function is said to be a \emph{trivial}, 
    as its Euler-Lagrange equations are identically satisfied, i.e.
    they are not difference equations.
    \label{rem:lequiv}
\end{remark}

Finally we give the following definition:
\begin{definition}
    A discrete Lagrangian $L=L\left( x_{n+k},\dots,x_{n} \right)$ is called \emph{normal} if
    \begin{equation}
        \pdv{L}{x_{n+k},x_{n}}\left( x_{n+k},\dots,x_{n} \right) \neq 0.
        \label{eq:lagrnormal}
    \end{equation}
    \label{def:normal}
\end{definition}

The importance of considering normal Lagrangians come from the fact that
non-normal Lagrangians define discrete Euler--Lagrange equations of order 
$2k-2$ at most. This is the content of the following lemma:
\begin{lemma}
    A non-normal discrete Lagrangian $L$ defined on $k$ points is equivalent
    to a normal discrete Lagrangian $\tilde{L}$ defined on $k-m$ points, 
    where $m\in\{1,\dots,k\}$ is the smallest integer such that:
    \begin{equation}
        \pdv{L}{x_{n+k-m},x_{n}}\left( x_{n+k},\dots,x_{n} \right) \neq 0.
        \label{eq:lagrabnormal}
    \end{equation}
    \label{lem:normal}
\end{lemma}

\begin{proof}
    Assume that the discrete Lagrangian $L=L\left( x_{n+k},\dots,x_{n} \right)$
    is non-normal.
    Then from definition \ref{def:normal} we have:
    \begin{equation}
       \pdv{L}{x_{n+k},x_{n}}\left( x_{n+k},\dots,x_{n} \right) = 0.
       \label{eq:lagrnormal1}
    \end{equation}
    Solving this partial differential equation and using the same
    reasoning as in the proof of Lemma \ref{lem:zeroder} we obtain
    that the Lagrangian $L$ has the following form:
    \begin{equation}
        \begin{aligned}
            L &= L_{1}\left( x_{n+k-1},\dots,x_{n} \right) 
            + \left(T_{n} -\Id  \right) f_{n}\left( x_{n+k-1},\dots,x_{n} \right)
            \\
            &\lequiv L_{1}\left( x_{n+k-1},\dots,x_{n} \right).
        \end{aligned}
        \label{eq:lagrnormalstep1}
    \end{equation}
    Therefore the original discrete Lagrangian $L$ is equivalent to
    a discrete Lagrangian $L_{1}$ defined on $k-1$ point. 
    The Euler--Lagrange equation of $L_{1}$ are defined on $2k-2$ at most.
    
    Now if:
    \begin{equation}
        \pdv{L}{x_{n+k-1},x_{n}}=0,
        \label{eq:lagrnormal2}
    \end{equation}
    we have that:
    \begin{equation}
        \pdv{L_{1}}{x_{n+k-1},x_{n}}=0,
        \label{eq:lagrnormal2bis}
    \end{equation}
    and the above reasoning can be iterated. Therefore, we obtain that 
    \begin{equation}
        L \lequiv L_{1}\left( x_{n+k-1},\dots,x_{n} \right) 
        \lequiv L_{2}\left( x_{n+k-2},\dots,x_{n} \right). 
        \label{eq:lagrnormalstep2}
    \end{equation}
    
    This procedure can be iterated until we reach $m\in\{1,\dots,k\}$ defined by
    the condition \eqref{eq:lagrabnormal}.
    Defining $\tilde{L}\equiv L_{m}$ the chain of equivalence gives us:
    \begin{equation}
        L\left( x_{n+k},\dots,x_{n} \right)
        \lequiv
        \tilde{L}\left( x_{n+k-m},\dots,x_{n} \right).
        \label{eq:equivm}
    \end{equation}
    This concludes the proof of the lemma.
\end{proof}

\begin{corollary}
    A discrete Lagrangian $L$ \eqref{eq:lagrgen} such that:
    \begin{equation}
        \pdv{L}{x_{n},x_{n+m}}=0, 
        \quad
        \forall m \in \left\{ 1,\dots,k \right\}
        \label{eq:lagrtrivial}
    \end{equation}
    is trivial.
    \label{cor:trivlagr}
\end{corollary}

\begin{proof}
    From Lemma \ref{lem:normal} we have $L\lequiv \tilde{L}\left( x_{n} \right)$,
    but $\tilde{L}\left( x_{n} \right)\lequiv \tilde{L}_{0}$ where $\tilde{L}_{0}$ 
    is a constant.
    Therefore $L\lequiv \tilde{L}_{0}$ and it is a trivial Lagrangian.
\end{proof}

Since in this paper we are interested in the inverse problem of the
discrete calculus of variations for equations of the form \eqref{eq:recgen}
depending on \emph{exactly} $2k$ points , from now on, 
we will always consider to deal with normal discrete Lagrangians.

In the next section we see which kind of condition must be satisfied to
guarantee the existence of a discrete Lagrangian for a given difference equation
of order $2k$ with $k>1$.

\section{Method for finding discrete Lagrangians}
\label{sec:method}

As we said in the Introduction, we want to solve the inverse problem 
of the discrete calculus of variations for scalar difference equations,
i.e. we want to be able to state when
the difference equation \eqref{eq:recgen} can be derived from a
Lagrangian \eqref{eq:lagrgen}.
To this end our first step is to state and prove following technical
lemma:
\begin{lemma}
    Consider a differential operator $\vec{A}$
    acting on the independent variables $x_{n+k-1}$, $x_{n+k-2}$,
    \dots, $x_{n-k}$, i.e.:
    \begin{equation}
        \vec{A} = \sum_{i=-k}^{k-1} A_{i} \pdv{}{x_{n+i}}, 
        \quad 
        A_{i}=A_{i}\left( x_{n+k-1},\dots,x_{n-k} \right),\, i=1\dots,k,
        \label{eq:gendiffop}
    \end{equation}
    such that \emph{for every} function 
    $G=G\left( x_{n+k},x_{n+k-1},\dots,x_{n} \right)$
    with $k>1$ 
    evaluated on the solutions of
    the difference equation \eqref{eq:recgen}
    \begin{equation}
        G=G\left(F\left( x_{n+k-1},x_{n+k-1},\dots,x_{n-k} \right) ,x_{n+k-1},\dots,x_{n} \right).
        \label{eq:Gdef}
    \end{equation}
    we have $\vec{A}\left( G \right) \equiv0$ identically.
    Then $\vec{A}$ is a linear combination of the following $k-1$ 
    differential operators:
    \begin{equation}
        \vec{A}^{+}_{m} = \pdv{F}{x_{n-m}}\pdv{}{x_{n-k}}- 
        \pdv{F}{x_{n-k}}\pdv{}{x_{n-m}},
        \quad
        m=1,\dots,k-1.
        \label{eq:annihil}
    \end{equation}
    \label{lem:annihil}
\end{lemma}

\begin{proof}
    Let us consider the operator in \eqref{eq:gendiffop} with unspecified
    coefficients $A_{i}$.
    We will fix the form of these coefficients by imposing the condition
    that the operator \eqref{eq:gendiffop} applied to any function
    $G$ as in \eqref{eq:Gdef} is identically zero. 
    Applying the operator \eqref{eq:gendiffop} applied to a function
    $G$ as in \eqref{eq:Gdef} and using the chain rule we obtain:
    \begin{equation}
            \vec{A}\left( G \right) =
            \left[
            \sum_{i=1}^{k} A_{-i} 
            \pdv{F}{x_{n-i}}
            \right]
            \pdv{G}{F}
            +\sum_{i=0}^{k-1} A_{i} 
            \left[
            \pdv{F}{x_{n+i}}
            \pdv{G}{F}
            +\pdv{G}{x_{n+i}}
        \right].
        \label{eq:gendiffappl}
    \end{equation}
    Let us assume for the moment that the partial derivatives of $G$ 
    \begin{equation}
        \pdv{G}{F},\pdv{G}{x_{n+k-1}},\dots,\pdv{G}{x_{n-k}}
        \label{eq:Gders}
    \end{equation}
    are \emph{functionally independent}, which is the most general case.
    Then as the condition  $\vec{A}\left( G \right)\equiv0$ in
    \eqref{eq:gendiffappl} must hold for all $G$ we can annihilate the 
    coefficients of the partial derivatives of $G$ \eqref{eq:Gders}.
    From \eqref{eq:gendiffappl} this implies that $A_{i}=0$ for 
    $i=0,\dots,k$.
    On the other hand from the derivative with respect to the
    first argument we obtain the following condition:
    \begin{equation}
            \sum_{i=1}^{k} A_{-i} 
            \pdv{F}{x_{n-i}}
            =0.
        \label{eq:hyperplane}
    \end{equation}
    Equation \eqref{eq:hyperplane} defines an hyperplane of dimension
    $k-1$ orthogonal to the gradient vector of $F$ with respect to
    the variables $\vec{x}^{-}=\left( x_{n-1},\dots,x_{n-k} \right)$:
    \begin{equation}
        \grad_{\vec{x}^{-}} F = \left( \pdv{F}{x_{n-1}},\dots,\pdv{F}{x_{n-k}} \right)^{T}.
        \label{eq:minusgrad}
    \end{equation}
    A basis for the orthogonal space $\left( \grad_{\vec{x}^{-}} F \right)^{\perp}$
    is given by the vectors:
    \begin{equation}
        \vec{v}_{1} =
        \begin{pmatrix}
            \pdv{F}{x_{n-k}}
            \\
            0
            \\
            0
            \\
            \vdots
            \\
            0
            \\
            -\pdv{F}{x_{n-1}}
        \end{pmatrix},
        \,
        \vec{v}_{2} =
        \begin{pmatrix}
            0
            \\
            \pdv{F}{x_{n-k}}
            \\
            0
            \\
            \vdots
            \\
            0
            \\
            -\pdv{F}{x_{n-2}}
        \end{pmatrix},
        \dots,
        \vec{v}_{k-1} =
        \begin{pmatrix}
            0
            \\
            0
            \\
            \\
            \vdots
            \\
            0
            \\
            \pdv{F}{x_{n-k}}
            \\
            -\pdv{F}{x_{n-k+1}}
        \end{pmatrix}.
        \label{eq:basevect}
    \end{equation}
    Inserting the components of this base into \eqref{eq:gendiffop}
    we obtain the vector fields \eqref{eq:annihil}.
    Since \eqref{eq:basevect} is a basis the thesis follows in the
    case the when the partial derivatives of $G$ \eqref{eq:Gders}
    are functionally independent.

    In the particular case when the partial derivatives of $G$
    \eqref{eq:Gders} are functionally dependent there exists a functional
    relation of the following kind:
    \begin{equation}
        Q\left( \frac{\partial G}{\partial F},\pdv{G}{x_{n+k-1}},\dots,\frac{\partial G}{\partial x_{n}} \right)=0,
        \label{eq:Qdef}
    \end{equation}
    with $Q$ a given function.
    Relation \eqref{eq:Qdef} means that $G$ solves a first-order 
    partial differential equation.
    This amounts to say that the function $G$ has the form:
    \begin{equation}
        G = G\left( I_{1},\dots,I_{k} \right), 
        \quad
        I_{j}=I_{j}\left( F,x_{n+k-1},\dots,x_{n} \right),
        \label{eq:Gdec}
    \end{equation}
    where the derivatives $\partial G/\partial I_{j}$ and
    functions $I_{j}$ are functionally independent and their partial
    derivatives are functionally independent.
    To each function $I_{j}$ in \eqref{eq:Gdec} the result holds because
    of the first part of the proof.
    By direct computation using the chain rule we obtain that the result 
    holds for $G$ itself.
    Therefore we obtain that the result holds even in the case when 
    the partial derivatives of $G$ \eqref{eq:Gders} are functionally dependent.
    This concludes the proof of the lemma.
\end{proof}

\begin{remark}
    In the proof of Lemma \ref{lem:annihil} only derivatives are involved.
    For this reason Lemma \ref{lem:annihil} generalises immediately
    to the case when the function $F$ in \eqref{eq:recgen} and $G$ in
    \eqref{eq:Gdef} depend explicitly on $n$.
    Therefore, to obtain the non-autonomous version of Lemma \ref{lem:annihil}
    just replace each occurrence of $F$ with $F_{n}$ and of $G$ with $G_{n}$.
    \label{rem:annihiln}
\end{remark}

\begin{remark}
    We underline that in order to prove Lemma \ref{lem:annihil} it
    is fundamental to assume that the differential operator $\vec{A}$
    \eqref{eq:gendiffop} \emph{annihilates identically} on \emph{all} 
    the functions of the form \eqref{eq:Gdef}.
    Indeed, it is always possible to find non-trivial differential operators
    of the form \eqref{eq:gendiffop} which annihilates a \emph{particular}
    function, even though its first derivatives are functionally independent.
    A simple example of this fact arises in dimension two and it 
    is given by radial functions. 
    That is, the function 
    $G\left( x,y \right) = g\left( x^{2}+y^{2} \right)$ is annihilated
    by the first-order linear differential operator:
    \begin{equation}
        \vec{A} = y \pdv{}{x}-x\pdv{}{y}.
        \label{eq:radialann}
    \end{equation}
    However, as is it known from the theory of first-order linear partial
    differential equations \cite{Arnold2004book}, the general solution
    of $\vec{A}\left( G \right)=0$ is given by radial functions.
    Therefore \emph{for every} non-radial function $H=H\left( x,y \right)$
    we will have $\vec{A}\left( H \right)\neq 0$.
    Therefore, the linear differential operator \eqref{eq:radialann}
    does not satisfy the hypotheses of Lemma \ref{lem:annihil}.
    \label{rem:notannihil}
\end{remark}

If the equation \eqref{eq:recgen} can be solved uniquely for $x_{n-k}$
the evolution of the difference equation is defined in both sides
and we can write:
\begin{equation}
    x_{n-k} = \widetilde{F}\left( x_{n+k},x_{n+k-1},\dots,x_{n-k+1} \right).
    \label{eq:recgeninv}
\end{equation}
\label{rem:rightoperator}
In this case we have a ``mirrored'' version of Lemma \eqref{lem:annihil}
expressed as following:
%
\begin{lemma}
    Consider a differential operator $\vec{\widetilde{A}}$
    acting on the independent variables $x_{n+k}$, $x_{n+k-2}$, \dots, $x_{n-k+1}$, 
    i.e.:
    \begin{equation}
        \vec{\widetilde{A}} = \sum_{i=-k+1}^{k} \widetilde{A}_{i} \pdv{}{x_{n+i}}, 
        \quad 
        \widetilde{A}_{i}=\widetilde{A}_{i}\left( x_{n+k},\dots,x_{n-k+1} \right),\, i=1\dots,k,
        \label{eq:gendiffopinv}
    \end{equation}
    such that \emph{for every function}
    $\widetilde{G}=\widetilde{G}\left( x_{n},x_{n-1},\dots,x_{n-k} \right)$ 
    where $k>1$, evaluated on the solutions of
    the difference equation \eqref{eq:recgeninv}
    \begin{equation}
        \widetilde{G}=\widetilde{G}\left( x_{n},x_{n-1},\dots,
        \widetilde{F}\left( x_{n+k},x_{n+k-1},\dots,x_{n-k+1} \right) \right)
        \label{eq:Gdefinv}
    \end{equation}
    we have $\vec{\widetilde{A}}\left( \widetilde{G} \right) \equiv0$.
    Then $\vec{\widetilde{A}}$ is a linear combination of the following
    $k-1$ differential operators:
    \begin{equation}
        \vec{A}^{-}_{m} = \pdv{\widetilde{F}}{x_{n+m}}\pdv{}{x_{n+k}}- 
        \pdv{\widetilde{F}}{x_{n+k}}\pdv{}{x_{n+m}},
        \quad
        m=1,\dots,k-1.
        \label{eq:annihil2}
    \end{equation}
    \label{lem:annihil2}
\end{lemma}

\begin{proof}
    Analogous to the proof of Lemma \ref{lem:annihil}.
\end{proof}

\begin{remark}
    Analogously as Lemma \ref{lem:annihil} Lemma \ref{lem:annihil2} 
    generalises immediately to the case when the function $\widetilde{F}$ 
    in \eqref{eq:recgeninv} and $\widetilde{G}$ in \eqref{eq:Gdefinv} 
    depend explicitly on $n$.
    To obtain the non-autonomous version of Lemma \ref{lem:annihil2}
    just replace each occurrence of $\widetilde{F}$ with $\tilde{F}_{n}$ 
    and of $\widetilde{G}$ with $\widetilde{G}_{n}$.
    \\
    Moreover, as for Lemma \ref{lem:annihil}, Lemma \ref{lem:annihil2}
    relies on the fact that the differential operator $\vec{\widetilde{A}}$ 
    in \eqref{eq:gendiffopinv} \emph{annihilates identically} on \emph{all}
    the functions of the form \eqref{eq:Gdefinv}.
    It is not possible to release such hypothesis, as the same counterexample
    presented in Remark \ref{rem:notannihil} is valid.
    \label{rem:annihiln2}
\end{remark}

Then we state the following definition:
\begin{definition}
    The differential operators introduced \eqref{eq:annihil}
    and \eqref{eq:annihil2} are called \emph{annihilation operators}.
    In particular the operators \eqref{eq:annihil} are the
    \emph{forward annihilation operators}, while the operators
    \eqref{eq:annihil2} are the \emph{backward annihilation operators}.
    \label{def:annihil}
\end{definition}

\begin{remark}
    The annihilation operators defined by equation \eqref{eq:annihil}
    and \eqref{eq:annihil2} are the one-dimensional analogous of 
    the operators $Y^{l}$ and $Z^{-l}$, for $l\in\Z$, defined in 
    \cite{LeviYamilov2011,Garifullin2012,GarifullinYamilov2012}.
    These operators have application also in the theory of Darboux
    integrable partial difference equations \cite{AdlerStartsev1999}.
    In \cite{GarifullinYamilov2012,GarifullinYamilov2015,GSY_DarbouxI}
    they were used to find the first integrals of some classes
    of partial difference equations.
    These operators annihilates all the dependent shifts of
    a quad equation, while $\vec{A}^{\pm}_{m}$ annihilates the dependent
    shifts of a scalar difference equation.
    \label{eq:remsym}
\end{remark}

Before going on it is important to note that the condition
$k>1$ in Lemmas \ref{lem:annihil} and \ref{lem:annihil2} cannot be relaxed.
In fact we can prove the following, complementary result:
\begin{lemma}
    Consider a first-order linear differential operator of the form
    \begin{equation}
        \vec{A} = \alpha\left( x_{n},x_{n-1} \right) \pdv{}{x_{n}}+
        \beta\left( x_{n},x_{n-1} \right)\pdv{}{x_{n-1}},
        \label{eq:genopdim2}
    \end{equation}
    such that \emph{for every} function $g=g\left( x_{n+1},x_{n} \right)$
    where $x_{n+1}$ solves a scalar second-order difference equation
    of the form $x_{n+1} = f\left( x_{n},x_{n-1} \right)$, i.e.:
    \begin{equation}
        g=g\left( f\left( x_{n},x_{n-1} \right),x_{n} \right).
        \label{eq:gsmalldef}
    \end{equation}
    we have $\vec{A}\left( g \right)\equiv 0$ identically.
    Then the linear differential operator $\vec{A}$ is trivial, 
    i.e. $\vec{A}\equiv 0$.

    Analogously, consider a first-order differential operator 
    of the form
    \begin{equation}
        \vec{\widetilde{A}} = \tilde{\alpha}\left( x_{n+1},x_{n} \right) \pdv{}{x_{n+1}}+
        \tilde{\beta}\left( x_{n+1},x_{n} \right)\pdv{}{x_{n}}
        \label{eq:genopdim2bis}
    \end{equation}
    such that \emph{for every} function
    $\tilde{g} = \tilde{g}\left( x_{n},x_{n-1} \right)$ 
    where $x_{n-1}$ solves a scalar second-order difference equation
    of the form $x_{n-1} = \tilde{f}\left( x_{n+1},x_{n} \right)$, i.e.:
    \begin{equation}
        \tilde{g}=\tilde{g}\left(x_{n}, \tilde{f}\left( x_{n+1},x_{n} \right) \right).
        \label{eq:gsmalldefinv}
    \end{equation}
    we have $\vec{\widetilde{A}}\left(\tilde{g} \right)\equiv 0$.
    Then the linear differential operator $\vec{\widetilde{A}}\equiv 0$ is trivial, 
    i.e. $\vec{\widetilde{A}}$.
    \label{lem:noannihil}
\end{lemma}

\begin{proof}
    Applying the operator \eqref{eq:genopdim2} to the
    function $g$ we obtain using the chain rule:
    \begin{equation}
        \begin{aligned}
            \vec{A}\left( g \right)&=
        \alpha \pdv{g}{x_{n}}\left( f\left( x_{n},x_{n-1} \right),x_{n} \right)+
        \beta\pdv{g}{x_{n-1}}\left( f\left( x_{n},x_{n-1} \right),x_{n} \right)
        \\
        &=\left[ \alpha\pdv{f}{x_{n}}+\beta\pdv{f}{x_{n-1}} \right]\pdv{g}{f}
            +\alpha\pdv{g}{x_{n}}.
        \end{aligned}
        \label{eq:appl}
    \end{equation}
    Assuming that the derivatives of $g$ with respect to its arguments
    are independent from the arbitrariness of $g$ we must annihilate 
    their coefficients.
    This implies that $\alpha=\beta\equiv0$.
    When the derivatives are no longer independent we can use the same
    argument as in Lemma \ref{lem:annihil} to reduce to the case when
    they are independent.

    Performing the same reasoning in the in case when $\tilde{g}$ is given
    by \eqref{eq:gsmalldefinv} and the general first-order differential
    operator $\vec{\widetilde{A}}$ \eqref{eq:genopdim2bis} we obtain 
    $\tilde{\alpha}=\tilde{\beta}\equiv0$ and the proof is completed.
\end{proof}

Using the annihilation operators $\vec{A}^{\pm}_{m}$ we can prove
our main result.
As this result relies on the general expression for the Euler--Lagrange
equation \eqref{eq:elgen}, we will formulate it in the non-autonomous
case, when such formula can be expressed explicit as equation
\eqref{eq:elgennonaut}.

\begin{theorem}
    Assume that there exists a discrete Lagrangian $L=L_{n}$ \eqref{eq:lagrgen}
    for equation \eqref{eq:recgen}, where $k>1$ and the right-hand side
    can depend explicitly on $n$, i.e. $F=F_{n}$.
    Then the Lagrangian $L_{n}$ satisfies the following
    system of linear partial differential equations:
    \begin{equation}
        \pdv{}{ x_{n-k}}
        \left\{
            \left(\pdv{F_{n}}{x_{n-k}}\right)^{-1}\vec{A}^{+}_{m}
            \left[\pdv{L_{n-k}}{x_{n}}\left( x_{n},\dots,x_{n-k} \right)\right]
        \right\}
        =0,
        \label{eq:lagrcondsys}
    \end{equation}
    where $m=1,\dots,k-1$.

    Moreover, if the equation \eqref{eq:recgen} can be solved uniquely 
    for $x_{n-k}$ and the evolution in the backward direction
    is given by equation \eqref{eq:recgeninv} where the right-hand side
    can depend explicitly on $n$, i.e. $\widetilde{F}=\widetilde{F}_{n}$
    then the discrete Lagrangian $L_{n}$
    \eqref{eq:lagrgen} satisfies the following
    system of linear partial differential equations:
    \begin{equation}
        \pdv{}{ x_{n+k}}
        \left\{
            \left(\pdv{\widetilde{F}_{n}}{x_{n+k}}\right)^{-1}\vec{A}^{-}_{m}
            \left[\pdv{L_{n}}{x_{n}}\left( x_{n+k},\dots,x_{n} \right)\right]
        \right\}
        =0,
        \label{eq:lagrcondsysinv}
    \end{equation}
    where $m=1,\dots,k-1$.
    \label{thm:elexist}
\end{theorem}

\begin{proof}
    Applying the annihilation operator $\vec{A}^{+}_{m}$ \eqref{eq:annihil} 
    to the Euler-Lagrange equation in explicit form \eqref{eq:elgennonaut} 
    we obtain: 
    \begin{equation}
        \sum_{l=0}^{k} \vec{A}^{+}_{m}
        \left[\pdv{L_{n-l}}{x_{n}}\left( x_{n+k-l},x_{n+k-1-l},\dots,x_{n-l} \right)\right] = 0.
        \label{eq:lagrder}
    \end{equation}
    Using the result of Lemma \eqref{lem:annihil}, as $x_{n+k}$ must be 
    evaluated along the solutions of equation \eqref{eq:recgen},
    and the definition of the annihilation operators \eqref{eq:annihil} 
    we have the following result:
    \begin{equation}
        \vec{A}^{+}_{m}\left[\pdv{L_{n-l}}{x_{n}}\left( x_{n+k-l},\dots,x_{n-l} \right)\right]
        =
        \begin{cases}
            0 & \text{if $l<m$},
            \\
            \displaystyle
            -\pdv{F_{n}}{x_{n-k}}\pdv{L_{n-l}}{x_{n},x_{n-m}}\left( x_{n+k-l},\dots,x_{n-l} \right)
            & \text{if $m\leq l \leq k-1$}
            \\
            \displaystyle
            \vec{A}^{+}_{m}\left[\pdv{L_{n-k}}{x_{n}}\left( x_{n},\dots,x_{n-k} \right)\right]
            & \text{if $l=k$}.
        \end{cases}
        \label{eq:opcases}
    \end{equation}
    Introducing this results in \eqref{eq:lagrder} we have:
    \begin{equation}
        \vec{A}^{+}_{m}\left[\pdv{L_{n-k}}{x_{n}}\left( x_{n},\dots,x_{n-k} \right)\right]
        =
        \pdv{F_{n}}{x_{n-k}}
        \sum_{l=m}^{k-1}
        \pdv{L_{n-l}}{ {x_{n-k-l}},{x_{n}} }\left( x_{n+k-l},\dots,x_{n-l} \right).
        \label{eq:lagrcond}
    \end{equation}
    Then dividing \eqref{eq:lagrcond} by $\partial F_{n}/\partial x_{n-k}$
    and differentiating with respect to $x_{n-k}$ we get the system \eqref{eq:lagrcondsys}.

    Reasoning in the same way with \eqref{eq:recgeninv} and the backward
    annihilation operators \eqref{eq:annihil2} we obtain the
    system \eqref{eq:lagrcondsysinv}.
    This concludes the proof.
\end{proof}

\begin{remark}
    We underline that the system \eqref{eq:lagrcondsys} 
    consist of \emph{at least} $k-1$ equations.
    In fact each equation in the system \eqref{eq:lagrcondsys}
    can be complemented by some additional conditions given by
    the fact that $F$ in \eqref{eq:recgen} depends
    on the variables $x_{n+i}$ for $i = -k, \dots, k -1$,
    while $L_{n-k}=L_{n-k}\left(x_{n},\dots,x_{n-k}\right)$.
    Therefore, we must require to each equation in \eqref{eq:lagrcondsys} to
    be independent of additional variables $x_{n+k-r}$, $r = 1, \dots, k -1$.
    This implies that the equations
    \begin{equation}
        \pdv{}{ {x_{n-k}},{x_{n+k-r}} }
        \left\{
            \left(\pdv{F_{n}}{x_{n-k}}\right)^{-1}
            \vec{A}^{+}_{m}\left[\pdv{L_{n-k}}{x_{n}}\left( x_{n},\dots,x_{n-k} \right)\right]
        \right\}
        =0,
        \label{eq:lagrcondsysder}
    \end{equation}
    for $m,r=1,\dots,k-1$ must hold.
    In the simpler, yet usual, case when  $F_{n}$ in \eqref{eq:recgen}
    is a rational function of its argument the $k-1$ conditions
    in \eqref{eq:lagrcondsysder} can be replaced by taking the numerator
    and then the coefficients with respect to to the variables 
    $x_{n+k-r}$ for $r = 1, \dots, k -1$ in equation \eqref{eq:lagrcondsys}.
\label{rem:coeffconf}
\end{remark}

\begin{remark}
    We remark that the result of Theorem \eqref{thm:elexist} do
    not require to the recurrence relation \eqref{eq:recgen} to 
    be defined in both directions, nor it requires to it to be
    rational.
    See Section \ref{sec:examples} for examples of non-rational
    recurrence relations admitting Lagrangians.
    \label{rem:birat}
\end{remark}

Theorem \ref{thm:elexist} is an effective tool to compute
Lagrangians or to prove if a given difference equation
of even order do not possess one.
In fact we can state the following:

\begin{corollary}
    Assume that the general solution of the system of linear partial
    differential equation \eqref{eq:lagrcondsys} (or \eqref{eq:lagrcondsysinv})
    associated with a given forward difference equation \eqref{eq:recgen}
    (or backward difference equation \eqref{eq:recgeninv}), with $k>1$, 
    gives raise to a trivial Lagrangian.
    Then there exists no non-trivial Lagrangian for 
    the forward difference equation \eqref{eq:recgen}
    (or for the backward difference equation \eqref{eq:recgeninv}).
    \label{cor:nolagr}
\end{corollary}

\begin{proof}
    Assume by contradiction that equation \eqref{eq:recgen} possess
    a non-trivial Lagrangian $L_{n}$.
    From Theorem \ref{thm:elexist} such Lagrangian will be
    a particular solution of the system \eqref{eq:lagrcondsys}.
    However, since the general solution of the linear system \eqref{eq:lagrcondsys}
    is trivial also $L_{n}$ must be trivial.
    This is a contradiction.
    Reasoning in the same way with \eqref{eq:recgeninv} and the
    system \eqref{eq:lagrcondsysinv} we obtain another contradiction.
    This concludes the proof.
\end{proof}

To practically the problem of finding a discrete Lagrangian $L=L_{n}$
we propose the following algorithmic steps.

\begin{enumerate}
    \item Introduce the function:
        \begin{equation}
            \ell_{n-k}\left( x_{n},\dots,x_{n-k} \right) = 
            \pdv{L_{n-k}}{x_{n}}\left( x_{n},\dots,x_{n-k} \right).
            \label{eq:elldef}
        \end{equation}
    \item Then the system \eqref{eq:lagrcondsys} becomes:
        \begin{equation}
            \pdv{}{ x_{n-k}}
            \left\{
                \left(\pdv{F_{n}}{x_{n-k}}\right)^{-1}\vec{A}^{+}_{m}
                \left[\ell_{n-k}\left( x_{n},\dots,x_{n-k} \right)\right]
            \right\}
            =0,
            \label{eq:lagrcondsysell}
        \end{equation}
        where $m=1,\dots,k-1$.
        That is the system \eqref{eq:lagrcondsys} reduces 
        to a linear system of second-order partial differential 
        equations for $\ell_{n-k}$. 
    \item The next step depends on the functional form of the
        the system \eqref{eq:lagrcondsysell}:
        \begin{enumerate}
            \item If the system \eqref{eq:lagrcondsysell} is rational in 
                $x_{n+k-r}$, $r = 1, \dots, k -1$ take the numerator of each
                equation and then the coefficients with respect to the variables
                $x_{n+k-r}$, $r = 1, \dots, k -1$.
                This gives the final system of linear partial differential
                equations to solve.
            \item If the system \eqref{eq:lagrcondsysell} is not rational
                consider the derivatives of each equation with respect to the
                variables $x_{n+k-r}$, $r = 1, \dots, k -1$:
                \begin{equation}
                    \pdv{}{ {x_{n-k}},{x_{n+k-r}} }
                    \left\{
                        \left(\pdv{F_{n}}{x_{n-k}}\right)^{-1}
                        \vec{A}^{+}_{m}\left[\ell_{n-k}\left( x_{n},\dots,x_{n-k} \right)\right]
                    \right\}
                    =0.
                    \label{eq:lagrcondsysderell}
                \end{equation}
                This gives the final system of linear partial differential
                equations to solve.
        \end{enumerate}
        In both cases solve the resulting final system using, e.g. 
        a computer algebra system like Maple or Mathematica.
    \item Recover the Lagrangian is recovered through integration and translation:
        \begin{equation}
            L_{n}\left( x_{n+k},\dots,x_{n} \right) = 
            \int^{x_{n+k}}\ell_{n}\left( y,x_{n+k-1},\dots,x_{n} \right)\ud y
            + \hat{\ell}_{n}\left( x_{n+k-1},\dots,x_{n} \right).
            \label{eq:ellint}
        \end{equation}
    \item Fix the form of the remaining arbitrary functions using
        the compatibility conditions given by equation \eqref{eq:lagrcond}.
    \item Check if the Euler--Lagrange equation \eqref{eq:elgennonaut} is 
        identically satisfied.
\end{enumerate}

If for technical reasons it is better to use the conditions
\eqref{eq:lagrcondsysinv}, we can perform the steps (i--vi)
using the following function:
\begin{equation}
    \tilde{\ell}_{n}\left( x_{n+k},\dots,x_{n} \right) = 
    \pdv{L_{n}}{x_{n}}\left( x_{n+k},\dots,x_{n} \right).
    \label{eq:elltdef}
\end{equation}
Instead of the system \eqref{eq:lagrcondsysell} we will solve
the analogous system of second-order linear
partial differential equations for $\tilde{\ell}_{n}$ obtained from
\eqref{eq:lagrcondsysinv}.
Again the Lagrangian will be recovered through the integration of 
equation \eqref{eq:elltdef} and the remaining arbitrary functions
will be fixed using the Euler--Lagrange themselves as compatibility
conditions.
However, we underline that, if needed, it is possible to mix the
forward and the backward approach in order to simplify the conditions.
See subsection \ref{sss:Qiii} for an example of such occurrence.

\begin{remark}
    In the case when $k=2$ we have that the systems \eqref{eq:lagrcondsys} 
    and \eqref{eq:lagrcondsysinv}
    actually consist of a \emph{single} equations:
    \begin{subequations}
    \begin{equation}
        \pdv{}{{x_{n-2}}}
    \left\{ 
            \left( \pdv{F_{n}}{x_{n-2}} \right)^{-1}
            \vec{A}^{+}_{1}\left[ \ell_{n-2}\left( x_{n},x_{n-1},x_{n-2} \right) \right] 
        \right\}.
        \label{eq:lagrcond2k2}
    \end{equation}
    and
    \begin{equation}
        \pdv{}{{x_{n+2}}}
    \left\{ 
        \left( \pdv{\widetilde{F}_{n}}{x_{n+2}} \right)^{-1}
        \vec{A}^{-}_{1}\left[ \tilde{\ell_{n}}\left( x_{n+2},x_{n+1},x_{n} \right) \right] 
        \right\}.
        \label{eq:lagrcond2k2inv}
    \end{equation}
        \label{eq:lagracond2k2sys}
    \end{subequations}
    If $F$ is rational to obtain the final system one needs to
    consider the coefficients with respect to $x_{n+1}$, otherwise
    to differentiate with respect to it in \eqref{eq:lagrcond2k2}.
    The same with respect to $\widetilde{F}$ and $x_{n-1}$ in \eqref{eq:lagrcond2k2inv}.
    \label{rem:k2case}
\end{remark}

In the next section we see some examples of the theory we presented above.

\section{Examples}
\label{sec:examples}

In this section we present several examples, both positive and
negative, of the application of the method discussed in the
previous one.

We mainly present autonomous examples.
In all autonomous examples the conditions presented in Theorem
\ref{thm:elexist} hold simply dropping the subscript $n$.
In principle, autonomous systems can have non-autonomous Lagrangians.
In the continuous case several instances of this fact are known
\cite{laggal,CP07Rao1JMP,AndersonThompson1992}.
All the examples of autonomous difference equations presented in 
this paper and known to the authors also have an autonomous Lagrangian.
Therefore, the problem of the finding a non-autonomous 
(up to equivalence)
discrete Lagrangian generating an autonomous difference equation is, 
up to our knowledge, open.
In the last subsection we present a non-autonomous example, 
the $\dP_{\text{I}}^{(2)}$ equation which shows how the method works 
in full generality.

Before going on, we note that in the working of the presented
examples we will employ extensively the computational techniques
and reasoning we introduced for the first time in the proofs 
of Lemmas \ref{lem:zeroder} and \ref{lem:normal}.

\subsection{The autonomous $\dP_{\text{II}}^{(2)}$ equation}

Let us consider the autonomous $\dP_{\text{II}}^{(2)}$ equation
\cite{JoshiViallet2017,CresswellJoshi1999}:
\begin{equation}
    \begin{aligned}
        x_{n+2} (1-x_{n+1}^2)+x_{n-2} (1-x_{n-1}^2)
        &=(x_{n+1}+x_{n-1}) \left[x_{n} (x_{n+1}+x_{n-1})+C\right]
        \\
        &-\frac{A x_{n}+B}{1-x_{n}^2}.
    \end{aligned}
    \label{eq:dPII2}
\end{equation}

We start then from equation \eqref{eq:lagrcond2k2} which,
when $F$ is given by solving \eqref{eq:dPII2} with respect to $x_{n+2}$, is:
\begin{equation}
    \begin{aligned}
    \left[  \left( 2 x_{{n-1}}+2 x_{{n+1}} \right) x_{{n}}+2 x_{{n-2}} x_{{n-1}}+C \right]
    &{\frac {\partial ^{2} \ell}{\partial x_{{n-2}}^{2}}}
    \\
    +\left( 1-x_{{n-1}}^{2} \right) {\frac {\partial ^{2} \ell}{\partial x_{{n-1}}\partial x_{{n-2}}}} 
        +2 x_{{n-1}}&{\frac {\partial \ell}{\partial x_{{n-2}}}} =0,
    \end{aligned}
    \label{eq:ccPII21}
\end{equation}
with $\ell = \ell\left( x_{n},x_{n-1},x_{n-2} \right)$.
Since there is no dependence on $x_{n+1}$ we can take coefficients
with respect and, upon solving with respect to higher order derivatives,
we obtain:
\begin{equation}
    {\frac {\partial ^{2} \ell}{\partial x_{{n-2}}^{2}}}=0,
    \quad
    {\frac {\partial ^{2} \ell}{\partial x_{{n-1}}\partial x_{{n-2}}}}=
    {\frac {2x_{{n-1}}}{{x_{{n-1}}}^{2}-1}}{\frac {\partial \ell}{\partial x_{{n-2}}}}
    \label{eq:ccfPII2}
\end{equation}
The solution of \eqref{eq:ccfPII2} is:
\begin{equation}
    \ell\left( x_{n},x_{n-1},x_{n-2} \right)
    =
    \left( x_{n-1}^{2}-1 \right)x_{n-2}\ell_{1}'\left( x_{n} \right)
    +\pdv{\ell_{2}}{x_{n}}\left( x_{n},x_{n-1} \right).
    \label{eq:ellPII2}
\end{equation}
Applying twice the translation operator in the positive direction
and integrating \eqref{eq:ellPII2} using \eqref{eq:ellint}
we obtain:
\begin{equation}
    \begin{aligned}
        L\left( x_{n+2},x_{n+1},x_{n} \right)
        &=
    \left( x_{n+1}^{2}-1 \right)x_{n}\ell_{1}\left( x_{n+2} \right)
    \\
    &+\ell_{2}\left( x_{n+2},x_{n+1} \right)
    +\ell_{3}\left( x_{n+1},x_{n} \right).
    \end{aligned}
    \label{eq:LdPII}
\end{equation}
Using the arbitrariness of $\ell_{2}\left( x_{n+2},x_{n+1} \right)$ 
we replace it by 
\begin{equation}
    \ell_{2}\left( x_{n+2},x_{n+1} \right)-\ell_{3}\left( x_{n+2},x_{n+1} \right).
\end{equation}
This allows to get rid of a total derivative in \eqref{eq:LdPII}
and we are left with the simplified Lagrangian:
\begin{equation}
    L\left( x_{n+2},x_{n+1},x_{n} \right)
    =
    \left( x_{n+1}^{2}-1 \right)x_{n}\ell_{1}\left( x_{n+2} \right)
    +\ell_{2}\left( x_{n+2},x_{n+1} \right).
    \label{eq:LdPIIbis}
\end{equation}
We must now check that the compatibility condition \eqref{eq:lagrcond}
is satisfied. Inserting \eqref{eq:LdPIIbis} in it we obtain:
\begin{equation}
    {\frac {\partial ^{2} \ell_{2}}{\partial x_{{n-1}}\partial x_{{n}}}}\left( x_{{n}},x_{{n-1}} \right)
    +2 x_{{n}} \ell_{1} \left( x_{{n+1}} \right)
    = 
    \left[  \left( 2 x_{{n-1}}+2 x_{{n+1}} \right) x_{{n}}+C \right] \ell_1' \left( x_{{n}} \right).
    \label{eq:ccPII22}
\end{equation}
Differentiating \eqref{eq:ccPII22} with respect to $x_{n+1}$ we obtain
\begin{equation}
    \ell_{1}'\left( x_{n} \right)=\ell_{1}'\left( x_{n+1} \right).
    \label{eq:ccPII221}
\end{equation}
This last equation implies that $\ell_{1}'\left( x_{n} \right)=K_{1}$ so that
$\ell_{1}\left( x_{1} \right)=K_{1}x_{n}$.
Substituting in \eqref{eq:ccPII22} we have:
\begin{equation}
    {\frac {\partial ^{2} \ell_{2}}{\partial x_{{n-1}}\partial x_{{n}}}}\left( x_{{n}},x_{{n-1}} \right) 
    = 
    K_{1}  \left( 2 x_{{n-1}}x_{{n}}+C \right).
    \label{eq:ccPII222}
\end{equation}
The solution of this last PDE is given by:
\begin{equation}
    \ell_{2}\left( x_{{n}},x_{{n-1}} \right)
    = 
    \frac{K_{1}}{2}x_{n}x_{n-1}\left( x_{n}x_{n-1}+2C \right)
    +\ell_{4}\left( x_{n} \right)
    +\ell_{5}\left( x_{n-1} \right).
    \label{eq:F3dPII2}
\end{equation}
As before we can replace $\ell_{4}\left( x_{n} \right)$ by 
$\ell_{4}\left( x_{n} \right)-\ell_{5}\left( x_{5} \right)$ and
inserting \eqref{eq:F3dPII2} and the value of $\ell_{1}$ 
into \eqref{eq:LdPIIbis} we obtain:
\begin{equation}
    L\left( x_{n+2},x_{n+1},x_{n} \right)
    =
    K_{1}\left[
    \left( x_{n+1}^{2}-1 \right)x_{n} x_{n+2} 
    +\frac{1}{2}x_{n+1}x_{n}\left( x_{n+1}x_{n}+2C \right)\right]
    +\ell_{4}\left( x_{n+1} \right)
    \label{eq:LdPII2}
\end{equation}
where we eliminated the total differences.
Now computing the Euler-Lagrange equations corresponding
to \eqref{eq:LdPII2} on the solutions of \eqref{eq:dPII2} we obtain:
\begin{equation}
    \ell_4' \left( x_{{n}} \right) =
    K_{1} {\frac { Ax_{{n}}+B}{ x_{{n}}^{2}-1}}.
    \label{eq:elccdPII2}
\end{equation}
Solving this last equation with respect to $\ell_{4}\left( x_{n} \right)$ we obtain:
\begin{equation}
    \ell_4(x_{n}) = \frac{K_{1}}{2}\left[(A+B)\log\left(x_{n}-1\right)+(A-B)\log\left(x_{n}+1\right)\right]
    +K_2.
    \label{eq:F5dPII2}
\end{equation}
The arbitrary constant is inessential to the Lagrangian, so
we can safely set $K_{2}=0$.
Finally, inserting \eqref{eq:F5dPII2} into \eqref{eq:LdPII2} 
and rescaling we obtain:
\begin{equation}
    \begin{aligned}
        L\left( x_{n+2},x_{n+1},x_{n} \right)
        &=
        \left( x_{n+1}^{2}-1 \right)x_{n} x_{n+2} 
        +\frac{1}{2}x_{n+1}x_{n}\left( x_{n+1}x_{n}+2C \right)
        \\
        &+ \frac{1}{2}\left[(A+B)\log\left(x_{n+1}-1\right)+(A-B)\log\left(x_{n+1}+1\right)\right].
    \end{aligned}
    \label{eq:LdPII2fin}
\end{equation}
This is the Lagrangian for equation \eqref{eq:dPII2}.

\subsection{The \eqref{eq:Qiii} equation}
\label{sss:Qiii}

Consider the following equation introduced in \cite{GJTV_class}:
\begin{equation}
    \begin{aligned}
    &\phantom{+}\left( x_{{n-2}}x_{{n-1}}^{2}-x_{{n+1}}^{2}x_{{n+2}}-Cx_{{n-1}}+Cx_{{n+1}} \right) x_{{n}}
    \\
    &+ \left(  \frac{C}{2}x_{{n-2}}+B \right) x_{{n-1}}
    - \left( \frac{C}{2} x_{{n+2}}+B \right) x_{{n+1}}
    =\frac{\alpha}{\beta} x_{{n}} \left( x_{{n+1}}-x_{{n-1}} \right).
    \end{aligned}
    \tag{Q.iii}
    \label{eq:Qiii}
\end{equation}
We will show that this equation has no Lagrangian.

We start then from equation \eqref{eq:lagrcond2k2},
which when $x_{n+2}$ is given by \eqref{eq:Qiii} is:
\begin{equation}
    \begin{aligned}
        &\phantom{+}
        \left[  \left( 4 x_{{n}}x_{{n-2}}x_{{n-1}}-2 Cx_{{n}}+Cx_{{n-2}}+2 B \right) \beta+2 \alpha x_{{n}} \right] 
         {\frac {\partial ^{2} \ell}{\partial x_{{n-2}}^{2}}}
         \\
         &- \left[ x_{{n-1}} \left( 2 x_{{n}}x_{{n-1}}+C \right) {
        \frac {\partial ^{2} \ell}{\partial x_{{n-1}}\partial x_{{n-2}}}} 
    - {\frac {\partial \ell}{\partial x_{{n-2}}}}  \left( 4 x_{{n}}x_{{n-1}}+C \right)  \right] \beta=0, 
    \end{aligned}
    \label{eq:ccQiii}
\end{equation}
where $\ell = \ell\left( x_{n},x_{n-1},x_{n-2} \right)$.
It is possible to solve directly equation \eqref{eq:ccQiii}, but the solution
is quite involved.
So instead of starting by solving this equation, we search for another
compatibility condition.
From equation \eqref{eq:lagrcond} differentiating with respect to
$x_{n+1}$ and translating backward once we obtain:
\begin{equation}
    \pdv{\ell}{ x_{n-1},x_{n-2} }=0.
    \label{eq:ccQiii2}
\end{equation}
Solving equations \eqref{eq:ccQiii2} and \eqref{eq:ccQiii} together
we obtain 
\begin{equation}
    \ell\left( x_{n},x_{n-1},x_{n-2} \right)= \ell_{1}\left( x_{n},x_{n-1} \right).
\end{equation}
Using the definition \eqref{eq:ellint} we have that
the possible Lagrangian for equation \eqref{eq:Qiii} is:
\begin{equation}
    L\left( x_{n+2},x_{n+1},x_{n} \right) = 
    \ell_{1}\left( x_{n+2},x_{n+1} \right)
    +\ell_{2}\left( x_{n+1},x_{n} \right).
    \label{eq:LQiii}
\end{equation}
This Lagrangian \eqref{eq:LQiii} is clearly non-normal.
As $\partial L / \partial x_{n+1}\partial x_{n}\neq0$, we have
from Lemma \ref{lem:normal} that Lagrangian \eqref{eq:LQiii} can 
define a difference equation of order two at most.
This is a contradiction, and we obtain that equation \eqref{eq:Qiii} 
does not possess a Lagrangian.
%
        
\subsection{Examples from \cite{TranvanderKampQusipel2016}}
\label{sub:tranex}

In \cite{TranvanderKampQusipel2016} the Ostrogradsky 
transformation was used to find Poisson structures for periodic reductions
of arbitrary order of four partial difference equations.
These Poisson structure were found exploiting the Lagrangian
formulation for those partial difference equations. 
Here we will show that these Lagrangians can be derived applying 
Theorem \ref{thm:elexist}.
Due to the algorithmic nature of the process we will present
the details of the derivation in one of the four cases.
All the other examples can be carried out in the same way, 
so we will not discuss them in details.

In \cite{TranvanderKampQusipel2016} the following four Lagrangians
of order $p+q$ are presented:
\begin{subequations}
    \begin{align}
        L_{\text{KdV}} &= x_{n}x_{n+p} - x_{n}x_{n+q} - \log\left( x_{n}-x_{n+p+q} \right),
        \label{eq:Lkdv}
        \\
        L_{\text{plusKdV}} &= x_{n}x_{n+p} + x_{n}x_{n+q} - \log\left( x_{n}+x_{n+p+q} \right),
        \label{eq:Lpluskdv}
        \\
        L_{\text{LV}} &= x_{n}x_{n+p} - x_{n}x_{n+q} - F\left( x_{n}-x_{n+p+q} \right),
        \label{eq:Llv}
        \\
        L_{\text{ASdL}} &= x_{n}x_{n+p} + x_{n}x_{n+q} - F\left( x_{n}+x_{n+p+q} \right),
        \label{eq:Llioville}
    \end{align}
    \label{eq:LTran}
\end{subequations}
where $p,q\in\N$ such that $p<q$ and $\gdc\left( p,q \right)=1$
and the function $F\left( x \right)$ is defined by the following integral:
\begin{equation}
    F\left( x \right) = \int_{0}^{x}\log\left( 1+e^{t} \right)\ud t.
    \label{eq:Fdef}
\end{equation}
These Lagrangian take their names from the fact that each of them
arises as $\left( p,-q \right)$ reduction of the discrete two dimensional
Lagrangian for the discrete KdV equation, the plus-KdV equation, the
discrete Lotka--Volterra equation \cite{LeviYamilov2009} and the 
Adler--Startsev discretization of the Liouville equation
\cite{AdlerStartsev1999} respectively.
The Euler--Lagrange equation of these Lagrangians are respectively:
\begin{subequations}
    \begin{gather}
        x_{n+p}+x_{n-p} -x_{n+q}-x_{n-q} - \frac{1}{x_{n}-x_{n+p+q}}+\frac{1}{x_{n-p-q}-x_{n}}=0,
        \label{eq:ELkdv}
        \\
        x_{n+p}+x_{n-p} -x_{n+q}-x_{n-q} - \frac{1}{x_{n}+x_{n+p+q}}-\frac{1}{x_{n-p-q}+x_{n}}=0,
        \label{eq:ELpluskdv}
        \\
        x_{{n+p}}+x_{{n-p}}-x_{{n+q}}-x_{{n-q}}
        +\log  \left( \frac{1+e^{x_{{n}}-x_{{n+q+p}}}}{1+e^{x_{{n-p-q}}-x_{{n}}}}\right) = 0, 
        \label{eq:ELlv}
        \\
        x_{{n+p}}+x_{{n-p}}+x_{{n+q}}+x_{{n-q}}
        +\log \left[ \left( 1+e^{x_{{n}}+x_{{n+q+p}}}\right)\left( 1+e^{x_{{n-p-q}}+x_{{ n}}} \right)\right] =0.
        \label{eq:ELliouville}
    \end{gather}
    \label{eq:ELTran}
\end{subequations}

We are going to prove the following result:
\begin{lemma}
    All the Lagrangians \eqref{eq:LTran} can be derived using the 
    conditions \eqref{eq:lagrcondsys} and \eqref{eq:lagrcondsysinv} 
    from Theorem \ref{thm:elexist}.
    \label{lem:lagrnihil}
\end{lemma}
\begin{proof}
    In the case of equation \eqref{eq:ELkdv} the forward annihilation 
    operators \eqref{eq:annihil} have the following expression:
    \begin{subequations}
        \begin{align}
            \vec{A}_{m}^{+} &= \frac{1}{D^{2}_{\text{KdV}}}
            \pdv{}{x_{n-m}},
            \quad m\neq p,q,
            \label{eq:Amgenkdv}
            \\
            \vec{A}_{p}^{+} &= \frac{1}{D^{2}_{\text{KdV}}}
            \left[\pdv{}{x_{n-p}}+ \left( x_{n}-x_{n-p-q} \right)^{2}\pdv{}{x_{n-p-q}}  \right],
            \label{eq:Ampkdv}
            \\
            \vec{A}_{q}^{+} &= \frac{1}{D^{2}_{\text{KdV}}} 
            \left[\pdv{}{x_{n-p}}-\left( x_{n}-x_{n-p-q} \right)^{2}\pdv{}{x_{n-p-q}}  \right],
            \label{eq:Amqkdv}
        \end{align}
        \label{eq:annihilopkdv}
    \end{subequations}
    where:
    \begin{equation}
        D_{\text{KdV}} = 1-\left( x_{n}-x_{n-p-q} \right)\left(x_{n+p}+x_{n-p}-x_{n+q}-x_{n-q}\right).
        \label{eq:Ddef}
    \end{equation}
    Using these operators we have that the differential
    conditions on the Lagrangian \eqref{eq:lagrcondsys} are given by:
    \begin{subequations}
        \begin{gather}
            \pdv{\ell}{x_{n-m},x_{n-p-q}} =0, \quad m\neq p,q,
            \label{eq:annkdvmgen}
            \\
            \begin{aligned}
            \left( x_{{n}}-x_{{n-q-p}} \right) ^{2}{\frac {\partial ^{2} \ell}{\partial x_{{n-q-p}}^{2}}}
            &+{\frac {\partial ^{2} \ell}{\partial x_{{n-q-p}}\partial x_{{n-p}}}}
            \\ 
            &+ 2\left( x_{{n-q-p}}-x_{{n}} \right) {\frac {\partial \ell}{\partial x_{{n-q-p}}}} =0,
            \end{aligned}
            \quad m=p,
            \label{eq:annkdvp}
            \\
            \begin{aligned}
            \left( x_{{n}}-x_{{n-q-p}} \right) ^{2}{\frac {\partial ^{2} \ell}{\partial x_{{n-q-p}}^{2}}}
            &-{\frac {\partial ^{2} \ell}{\partial x_{{n-q-p}}\partial x_{{n-q}}}}
            \\
            &+ 2\left( x_{{n-q-p}}-x_{{n}} \right) {\frac {\partial \ell}{\partial x_{{n-q-p}}}} =0,
            \end{aligned}
            \quad m=q.
            \label{eq:annkdvq}
        \end{gather}
        \label{eq:annkdv}
    \end{subequations}

    Now we prove that solving the conditions \eqref{eq:annkdv} we find
    the Lagrangian \eqref{eq:Lkdv}.
    First of all, we note that the equation \eqref{eq:ELkdv}
    depends only on the seven points $x_{n}$, $x_{n\pm p}$, $x_{n\pm q}$ and
    $x_{n\pm\left( p+q \right)}$.
    Therefore, we can make the simplifying assumption:
    \begin{equation}
        L = L \left( x_{n+p+q},x_{n+q},x_{n+p},x_{n} \right).
        \label{eq:Lkdvsimp}
    \end{equation}
    With this assumption the first set of equation \eqref{eq:annkdvmgen}
    is identically satisfied.
    Then the solution of equations \eqref{eq:annkdvp} and \eqref{eq:annkdvq}
    is given by:
    \begin{equation}
        \ell
        =
        \ell_{2}(x_{n-q}, x_{n-p}, x_{n})
        +\int
        \frac{\ell_{1}
        \left(x_{n}, x_{n-q}+\dfrac{1+(x_{n-p-q}-x_{n})x_{n-p}}{x_{n}-x_{n-p-q}}\right)}{(x_{n-p-q}-x_{n})^2}\ud x_{n-p-q}.
        \label{eq:ellkdvsol}
    \end{equation}
    Following the definition \eqref{eq:elldef} we obtain that a
    Lagrangian for equation \eqref{eq:ELkdv} must have the following
    form:
    \begin{equation}
        \begin{aligned}
            L(x_{n+p+q}, x_{n+q}, x_{n+p}, x_{n})
            &=
        \ell_{2}(x_{n+p}, x_{n+q}, x_{n})
        \\
        &+\iint\frac{\ell_{1}
        \left(x_{n+p+q}, x_{n+p}+\dfrac{1+(x_{n}-x_{n+p+q})x_{n+q}}{x_{n+p+q}-x_{n}}\right)}{(x_{n}-x_{n+p+q})^2}
        \ud x_{n}\ud x_{n+p+q}.
        \end{aligned}
        \label{eq:Lkdvsol}
    \end{equation}
    Here we used the arbitrariness of $\ell_{2}$ to change it to
    $\partial\ell_{2}/\partial x_{n}$ in order to keep the expression
    for $L$ as simple as possible.
    Deriving the Euler--Lagrange equations of the Lagrangian
    \eqref{eq:Lkdvsol}, applying the operator $\vec{A}_{p}^{+}$ \eqref{eq:Ampkdv}
    then differentiating with respect to $x_{n+p}$ we obtain:
    \begin{equation}
        \pdv{\ell_{1}}{\xi} = 0, 
        \quad 
        \xi = x_{n+p}+\frac{1+(x_{n}-x_{n+p+q})x_{n+q}}{x_{n+p+q}-x_{n}}.
        \label{eq:l1cond}
    \end{equation}
    This greatly simplifies the expression in \eqref{eq:Lkdvsol} to:
    \begin{equation}
        L(x_{n+p+q}, x_{n+q}, x_{n+p}, x_{n})
        =
        \ell_{2}(x_{n+p}, x_{n+q}, x_{n})
        -\int\frac{\ell_{1}\left(x_{n+p+q}\right)}{x_{n}-x_{n+p+q}}\ud x_{n+p+q}.
        \label{eq:Lkdvsol2}
    \end{equation}
    Computing the Euler--Lagrange equations of \eqref{eq:Lkdvsol2} then
    applying the operators $\vec{A}_{p}^{+}$ \eqref{eq:Ampkdv} 
    and $\vec{A}_{q}^{+}$ \eqref{eq:Amqkdv} we obtain
    that the function $\ell_{2}$ must satisfy the following system of
    PDEs:
    \begin{subequations}
        \begin{align}
        {\frac {\partial ^{2}\ell_{2} }{\partial x_{{n-p}}\partial x_{{n}}}}\left( x_{{n-q}},x_{{n-p}},x_{{n}} \right)
        &=-\ell_1 \left( x_{{n}}\right),
        \\
        {\frac {\partial ^{2} \ell_{2}}{\partial x_{{n-q}}\partial x_{{n}}}} \left( x_{{n-q}},x_{{n-p}},x_{{n}} \right) 
        &=\ell_{1} \left( x_{{n}}\right).
        \end{align}
        \label{eq:l2cond}
    \end{subequations}
    Introducing $L_{1}\left( x_{n} \right) = \int \ell_{1}\left( x_{n} \right)\ud x_{n}$
    and solving \eqref{eq:l2cond} 
    we obtain the following form for the Lagrangian of \eqref{eq:ELkdv}:
    \begin{equation}
        \begin{aligned}
        L(x_{n+p+q}, x_{n+q}, x_{n+p}, x_{n})
        &=
        \ell_3\left(x_{n+p}, x_{n+q}\right)
        +\left(x_{n+p}-x_{n+q}\right) L_1(x_{n})
        \\
        &+\ell_4(x_{n+p+q})+\ell_{5}(x_{n+q})
        -\int\frac{L_{1}'\left(x_{n+p+q}\right)}{x_{n}-x_{n+p+q}}\ud x_{n+p+q}.
        \end{aligned}
        \label{eq:Lkdvsol3}
    \end{equation}
    From the arbitrariness of $\ell_{4}\left( x_{n+p+q} \right)$
    we can write 
    \begin{equation}
        \ell_{4}\left( x_{n+p+q} \right)=
        \tilde{\ell}_{4}\left( x_{n+p+q} \right)-\ell_{5}\left( x_{n+p+q} \right)
    \end{equation}
    and since
    \begin{equation}
        \ell_{5}\left( x_{n+p+q} \right) - \ell_{5}\left( x_{n+p} \right)
        = \sum_{l=0}^{q}\left[ \ell_{5}\left( x_{n+p+l+1} \right) - \ell_{5}\left( x_{n+p+l} \right) \right]
        \lequiv 0,
        \label{eq:ell45td}
    \end{equation}
    we have the following simplification in \eqref{eq:Lkdvsol3}:
    \begin{equation}
        \begin{aligned}
        L(x_{n+p+q}, x_{n+q}, x_{n+p}, x_{n})
        &=
        \ell_3\left(x_{n+p}, x_{n+q}\right)
        +\left(x_{n+p}-x_{n+q}\right) L_1(x_{n})
        \\
        &+\ell_4(x_{n+p+q})
        -\int\frac{L_{1}'\left(x_{n+p+q}\right)}{x_{n}-x_{n+p+q}}\ud x_{n+p+q}.
        \end{aligned}
        \label{eq:Lkdvsol4}
    \end{equation}
    In \eqref{eq:Lkdvsol4} for sake of simplicity we dropped the tilde 
    in $\ell_{4}\left( x_{n+p+q} \right)$.
    From the Lagrangian \eqref{eq:Lkdvsol4} we obtain the following
    Euler--Lagrange equation:
    \begin{equation}
        \begin{aligned}
        &\phantom{+}\int 
        \frac{L_{1}' \left( x_{{n+q+p}} \right) }{\left( x_{{n}}-x_{{n+q+p}} \right) ^{2}}
        \ud x_{{n+q+p}}
        +{\frac {\partial \ell_{3}}{\partial x_{{n}}}}\left( x_{{n+p-q}},x_{{n}} \right) 
        -L_{1} \left( x_{{n+p}} \right) 
        \\
        &+{\frac {\partial \ell_{3}}{\partial x_{{n}}}}\left( x_{{n}},x_{{n+q-p}} \right) 
        +L_{1} \left( x_{{n+q}} \right) 
        + \left( x_{{n-q}}-x_{{n-p}} \right) L_{1}' \left( x_{{n}} \right) 
        \\
        &+\ell_4' \left( x_{{n}} \right) 
        +\frac{ L_1' \left( x_{{n}} \right)}{ x_{{n}}-x_{{n-q-p}}}=0.
        \end{aligned}
        \label{eq:ELkdv4}
    \end{equation}
    Confronting \eqref{eq:ELkdv4} with \eqref{eq:ELkdv} it is natural
    to make the linear \emph{ansatz} for $L_{1}\left( x_{n} \right)$, i.e.
    $L_{1}\left( x_{n} \right)=K x_{n}$ with $K$ a constant\footnote{Alternatively,
    without making ansatz it is possible to solve \eqref{eq:ELkdv} for
    $x_{n-p-q}$ and apply the backward annihilation operators \eqref{eq:annihil2}.}.
    Using this ansatz the Euler--Lagrange equation \eqref{eq:ELkdv4}
    becomes:
    \begin{equation}
        \begin{aligned}
        &\phantom{+}
        \frac{K}{x_{{n}}-x_{{n+q+p}}}
        +{\frac {\partial \ell_{3}}{\partial x_{{n}}}}\left( x_{{n+p-q}},x_{{n}} \right) 
        +{\frac {\partial \ell_{3}}{\partial x_{{n}}}}\left( x_{{n}},x_{{n+q-p}} \right)
        \\
        &+ K\left(x_{{n+q}}+ x_{{n-q}}-x_{n+p} -x_{{n-p}} \right)
        +\ell_4' \left( x_{{n}} \right) 
        +\frac{K}{ x_{{n}}-x_{{n-q-p}}}=0.
        \end{aligned}
        \label{eq:ELkdv4bis}
    \end{equation}
    Since $x_{n+p+q}$ do not depend on $x_{n+p-q}$ and $x_{n+q-p}$
    we have that $\partial^{2} \ell_{3}/\partial x_{n+p}\partial x_{n+q}=0$,
    i.e.
    \begin{equation}
        \ell_{3}\left( x_{n+p},x_{n+q} \right)=
        \ell_{6}\left( x_{n+p} \right)+\ell_{7}\left( x_{n+q} \right).
    \end{equation}
    However, exploiting again the arbitrarily of $\ell_{4}\left( x_{n+p+q} \right)$
    we can write 
    \begin{equation}
        \ell_{4}\left( x_{n+p+q} \right)=
        \tilde{\ell}_{4}\left( x_{n+p+q} \right)
        -\ell_{6}\left( x_{n+p+q} \right)-\ell_{7}\left( x_{n+p+q} \right).
    \end{equation}
    Then from a reasoning analogous to the one in formula
    \eqref{eq:ell45td} we can remove all the terms in $\ell_{6}$ and
    in $\ell_{7}$.
    Therefore we are left with the following Euler--Lagrange equations:
    \begin{equation}
        \frac{K}{x_{{n}}-x_{{n+q+p}}}
        + K\left(x_{{n+q}}+ x_{{n-q}}-x_{n+p} -x_{{n-p}} \right)
        +\ell_4' \left( x_{{n}} \right) 
        +\frac{K}{ x_{{n}}-x_{{n-q-p}}}=0.
        \label{eq:ELkdv4ter}
    \end{equation}
    Confronting with \eqref{eq:ELkdv} we have $\ell_{4}'\left( x_{n} \right)=0$
    so that $\ell_{4}\left( x_{n} \right) = \text{constant}$.
    Since constants are inessentials in Lagrangian we can safely put
    $\ell_{4}\left( x_{n} \right)=0$.
    Putting these considerations together we obtained the Lagrangian:
    \begin{equation}
        L(x_{n+p+q}, x_{n+q}, x_{n+p}, x_{n})
        =K\left(x_{n+p}-x_{n+q}\right) x_{n}
        +K\log\left(x_{n}-x_{n+p+q}\right),
        \label{eq:Lkdvsol5}
    \end{equation}
    which is equivalent to \eqref{eq:Lkdv} if we choose $K=-1$.

    In \ref{app:lemmaproof} we present the conditions 
    \eqref{eq:lagrcondsys} for the other three equations in \eqref{eq:ELTran}.
    The proof that the also the other three Lagrangians in \eqref{eq:LTran}
    can be obtained solving these conditions reported in \ref{app:lemmaproof}
    runs in the same way and therefore it is omitted.
\end{proof}

\subsection{The $\dP_{\text{I}}^{(2)}$ equation}
\label{sub:dPI2}

In this subsection we consider the following non-autonomous
equation:
\begin{equation}
    \begin{aligned}
        &x_{{n}}\left(x_{{n+1}}x_{{n+2}}+x_{{n-1}}x_{{n-2}}\right)
        +x_{n}x_{{n-1}}x_{{n+1}}
        \\
        +&2 x_{{n}}^{2}\left( x_{{n+1}}+x_{{n+1}}\right)
        +x_{{n}} \left( x_{{n+1}}^{2}+x_{{n}}^{2}+x_{{n-1}}^{2} \right) 
        \\
        +&c_{{3}}x_{{n}} \left( x_{{n-1}}+x_{{n}}+x_{{n+1}} \right) 
        +c_{{2}}x_{{n}}=c_{{1}}+c_{{0}} \left( -1 \right) ^{n}-n.
    \end{aligned}
    \label{eq:dPI2}
\end{equation}
This equation is the second member of the Painlev\'e I hierarchy,
in short $\dP_{\text{I}}^{(2)}$, and was first derived
in \cite{CresswellJoshi1999}.
We will now present a derivation of a discrete Lagrangian for
such equation.


First, we consider equation \eqref{eq:lagrcond2k2},
which when $F_{n}$ is given by solving \eqref{eq:dPI2} with respect to
$x_{n+2}$ is the following:
\begin{equation}
    \begin{aligned}
        \left[ 2 \left(x_{{n}}+x_{{n-1}}\right)+x_{{n+1}}+c_{{3}}+x_{{n-2}} \right]
        &{\frac {\partial ^{2}\ell_{n-2}}{\partial x_{{n-2}}^{2}}}
        \\
    -x_{{n-1}}{\frac {\partial ^{2}\ell}{\partial x_{{n-1}}\partial x_{{n-2}}}}
    +&{\frac {\partial \ell_{n-2}}{\partial x_{{n-2}}}}=0.
    \end{aligned}
    \label{eq:a1dPI2}
\end{equation}
Taking the coefficients with respect to $x_{n+1}$ in \eqref{eq:a1dPI2}
and solving the resulting system of partial differential equations
we obtain that $\ell_{n-2}$ has the following form:
\begin{equation}
    \ell_{n-2} = x_{n-2}x_{n-1}\ell_{1,n-2}(x_{n})+\ell_{2,n-2}(x_{n}, x_{n-1}).
    \label{eq:elldPI2}
\end{equation}
Using the definition \eqref{eq:elldef} and exploiting the
arbitrariness of $\ell_{2,n-2}$ we obtain the following expression
for a possible Lagrangian of equation \eqref{eq:dPI2}:
\begin{equation}
    L_{n}\left( x_{n+2},x_{n+1},x_{n} \right) = x_{n}x_{n+1}\ell_{1,n}(x_{n+2})
    +\ell_{2,n}(x_{n+2}, x_{n+1}).
    \label{eq:L1dPI2}
\end{equation}
Computing the Euler--Lagrange equations of \eqref{eq:L1dPI2} 
then applying the operator $\vec{A}_{1}^{+}$ we obtain the following
compatibility condition:
\begin{equation}
    {\frac {\partial ^{2} \ell_{2,n-2}}{\partial x_{{n-1}}\partial x_{{n}}}} 
    \left( x_{{n}},x_{{n-1}} \right)
    =\left( c_{{3}}+2 x_{{n}}+2 x_{{n-1}}+x_{{n+1}} \right) \ell_{1,n-2}' \left( x_{{n}} \right) 
    -\ell_{1,n-1} \left( x_{{n+1}} \right)
    \label{eq:cc1dPI2}
\end{equation}
Differentiating \eqref{eq:cc1dPI2} with respect to $x_{n+1}$
we obtain:
\begin{equation}
    \ell_{1,n-1}'\left( x_{n+1} \right) = \ell_{1,n-2}'\left( x_{n} \right).
    \label{eq:cc1dxpdPI2}
\end{equation}
We have that the functional equation \eqref{eq:cc1dxpdPI2} implies
$\ell_{1,n-2}\left( x_{n} \right)=K_{n-2}x_{n}$ where $K_{n-2}$
is a function depending explicitly on $n$.
Inserting such value of $\ell_{1,n-2}\left( x_{n} \right)$
into \eqref{eq:cc1dxpdPI2} we have $K_{n-1}=K_{n-2}$, i.e.
$K_{n}=K$ a constant for all $n\in\Z$.
So in the end we have $\ell_{1,n-2}\left( x_{n} \right) = K x_{n}$.
Substituting this value of $\ell_{1,n-2}\left( x_{n} \right)$ in equation
\eqref{eq:cc1dPI2} we and solving the obtained partial difference
equation with respect to $\ell_{2,n-2}$ we have:
\begin{equation}
    \ell_{2,n-2}\left( x_{{n}},x_{{n-1}} \right) =
    Kx_{{n}}x_{{n-1}} \left( c_{{3}}+x_{{n}}+x_{{n-1}} \right)
    +\ell_{3,n-2} \left( x_{{n}} \right) +\ell_{4,n-2} \left( x_{{n-1}} \right).
    \label{eq:ell2soldPI2}
\end{equation}
Now due to the arbitrariness of $\ell_{3,n-2}\left( x_{n} \right)$ we
can write 
\begin{equation}
    \ell_{3,n-2}\left( x_{n} \right)=\tilde{\ell}_{3,n-2}\left( x_{n} \right)
    -\ell_{4,n-1}\left( x_{n} \right)
\end{equation}
and, since $\ell_{4,n-1}\left( x_{n} \right)-\ell_{4,n-2}\left( x_{n-1} \right)\lequiv0$,
we obtain the following Lagrangian:
\begin{equation}
    L_{n}\left( x_{n+2},x_{n+1},x_{n} \right) = 
    K x_{n}x_{n+1}x_{n+2}
    +K x_{{n+2}}x_{{n+1}} \left( c_{{3}}+x_{{n+2}}+x_{{n+1}} \right)
    +\ell_{3,n} \left( x_{{n+2}} \right),
    \label{eq:L2dPI2}
\end{equation}
where for sake of simplicity we dropped the tilde in $\ell_{3,n}$.
Computing the Euler--Lagrange equations of \eqref{eq:L2dPI2} and
substituting the value of $x_{n+2}$ from \eqref{eq:dPI2} we
are left with the following condition:
\begin{equation}
    x_{n}\ell_{3,n-2}'\left( x_{n} \right)=
    K \left[ c_{{3}}x_{{n}}^{2}+x_{{n}}^{3}-c_{{0}} \left( -1\right) ^{n}+c_{{2}}x_{{n}}+n-c_{{1}} \right]. 
    \label{eq:cc3dPI2}
\end{equation}
Integrating equation \eqref{eq:cc3dPI2} and rescaling
to eliminate the inessential constant $K$
we obtain the following Lagrangian for equation \eqref{eq:dPI2}:
\begin{equation}
    \begin{aligned}
        L_{n}\left( x_{n+2},x_{n+1},x_{n} \right) &= 
    x_{n}x_{n+1}x_{n+2}
    +x_{{n+2}}x_{{n+1}} \left( c_{{3}}+x_{{n+2}}+x_{{n+1}} \right)
    \\
    &+\frac{x_{{n+2}}^{3}}{3}+\frac{c_{{3}}x_{{n+2}}^{2}}{2}+c_{{2}}x_{{n+2}}
    \\
    &+ \left[ n+2 -c_{0}\left( -1 \right) ^{n} -c_{{1}} \right] \log  x_{{n+2}}. 
    \end{aligned}
    \label{eq:L3dPI2}
\end{equation}
This ends the proof.

\section{Conclusions}
\label{sec:concl}

In this paper we discussed the conditions for the existence
of a discrete Lagrangian in the case of scalar difference equations
of arbitrary even order $2k$ with $k>1$.
Our main result is contained in
Theorem \ref{thm:elexist} and Corollary \ref{cor:nolagr}
which gives a way of computing such discrete Lagrangians.

The usefulness of the method presented in this paper extends also
to differential equations.
Indeed, from a difference equation of order $2k$ through the so-called 
\emph{continuum limit} procedure we can always obtain a differential equation of
order lesser or equal than $2k$.
Applying the same continuum limit to the discrete Lagrangian of such
difference equation we obtain a continuous Lagrangian for the corresponding
differential equation.
To give an example of this occurrences, let us consider the 
autonomous $\dP_{\text{II}}^{(2)}$ equation \eqref{eq:dPII2}.
It was proved  in \cite{CresswellJoshi1999} that the 
autonomous $\dP_{\text{II}}^{(2)}$ equation \eqref{eq:dPII2}
under the following scaling
\begin{subequations}
    \begin{gather}
        x_{n} = h u\left( t \right), \quad t = n h,
        \\
        A= 6  +2\alpha h^2+\delta h^4, \,
        B= \beta h^{5}, \,
        C= 4+\alpha h^{2},
    \end{gather}
    \label{eq:PII2scaling}
\end{subequations}
in the limit $h\to0$ reduces to the following differential
equation:
\begin{equation}
    u\lagrangeprime{4}    
    - \alpha u\lagrangeprime{2}  
    -10 u^{2} u\lagrangeprime{2} 
    + \left[ \delta-10  \left( u\lagrangeprime{1}  \right)^{2} \right] u
    +6u^{5}+2 \alpha u^{3}
    +\beta=0. 
    \label{eq:PII2}
\end{equation}
This equation is a translated version of the fourth-order member of
the $P_{\text{II}}$ hierarchy, the $P_{\text{II}}^{(2)}$ equation, introduced
in \cite{FlaschkaNewell1980}.
Applying the scaling \eqref{eq:PII2scaling} to the discrete Lagrangian
\eqref{eq:LdPII2fin} in the limit $h\to0$ we obtain: 
\begin{equation}
    \begin{aligned}
        \frac{L}{h^{6}} &= 
        \frac{1}{2} \left( u\lagrangeprime{2} \right) ^{2}
        -\frac{1}{2} u\left(u^{2}-\alpha \right)  u\lagrangeprime{2}
        +\left( \alpha+\frac{7}{2}u^{2} \right)\left(u\lagrangeprime{1}  \right)^{2}
        \\
        &+\frac{1}{2} u  \left( 2 u^{5}+ u^{3}\alpha+\delta u +2 \beta \right) 
        +O\left( h \right),
    \end{aligned}
    \label{eq:LagrPII2}
\end{equation}
up to the addition of a total derivative.
The leading order term in \eqref{eq:LagrPII2} is
a Lagrangian for the $P_{\text{II}}^{(2)}$ equation \eqref{eq:PII2}.
As the $P_{\text{II}}^{(2)}$ equation \eqref{eq:PII2} is a fourth-order
equations we obtain that the Lagrangian in \eqref{eq:LagrPII2}
is then its unique Lagrangian \cite{Fels1996}.

On the other hand the converse procedure,
i.e. going from continuous Lagrangian equations to discrete ones,
is not possible in general.
This happens because, in general, discretising
a differential equation does not preserves its variational structure.
The non-trivial problem of finding discretisation with variational structure
is part of the so-called  \emph{geometric integration theory} 
\cite{BuddIserles1999,BuddPiggot2003,KrantzParks2008}.
We believe that the method presented in this paper can be helpful in the
framework of the geometric integration theory.
Indeed, the algorithmic test presented in this paper can be used 
to isolate Lagrangian difference equations out of families of difference equations
with the same continuum limit.

The case $k=1$ seems to be more involved and less algorithmic
due to the result of Lemma \ref{lem:noannihil}.
It is possible that this reflects the well-known fact that
second-order differential equations 
\begin{equation}
    u\lagrangeprime{2} = F\left( t,u,u\lagrangeprime{1} \right),
    \label{eq:secondeq}
\end{equation}
admit infinitely many Lagrangians.
These Lagrangians $L=L(t,u,u\lagrangeprime{1})$ are generated 
by the Jacobi Last Multiplier solving the trivial 
partial differential equation:
\begin{equation}
    \pdv{L}{*{2}{{u\lagrangeprime{1}}}} = M
    \label{eq:jacobilagr}
\end{equation}
where $M$ is a Jacobi Last Multiplier of the second-order
differential equation \eqref{eq:secondeq} 
\cite{Jacobi44a,Jacobi45,Jacobi1846,nuctam_1lag,nuctam_2lag,nuctam_3lag}.
In \cite{LeviRodriguez2012} the Jacobi Last Multiplier was defined
for discrete equation, but no relationship with the discrete
Lagrangian was given.
Further investigations to establish the existence of a formula
analogous to \eqref{eq:jacobilagr} are needed.

Let us now comments the relationship of our solution of the
inverse problem of the discrete calculus of variation with
the one given in \cite{HydonMansfield2004}.
As we noticed in the Introduction in \cite{HydonMansfield2004} an 
analogous of the homotopy formula for discrete equations was introduced.
The homotopy formula presented in \cite{HydonMansfield2004} can be used 
to construct a discrete Lagrangian for scalar difference equations \eqref{eq:recgen}
for all $k\in\N$ as well as for systems of difference equations, 
and even in the case of partial difference equations.
However, there are some restrictions on the applicability of the 
homotopy formula presented in \cite{HydonMansfield2004}.

To keep the discussion simple we give a comparison between 
the applicability of the method presented in \cite{HydonMansfield2004} 
and our method.
First of all, let us state the results of \cite{HydonMansfield2004} in
our setting.
Let $P=P[x_{n}]$ be a difference function, i.e. a function depending
on $x_{n}$ and its shifts up to some order $r$.
We denote the space of such functions by $\mathcal{A}^{r}$.
The \emph{Fr\'echet derivative} of $P$ is
an operator $\vec{D}_{P}\colon \mathcal{A}^{q}\to\mathcal{A}^{r}$
defined by:
\begin{equation}
    \vec{D}_{P}\left( Q \right)
    =
    \lim_{\varepsilon\to0}\frac{P\left[ x_{n}+\varepsilon Q\left[ x_{n} \right] \right]-P\left[ x_{n} \right]}{\varepsilon}
    \label{eq:frechet}
\end{equation}
where $Q=Q\left[ x_{n} \right]$ is an arbitrary element of $\mathcal{A}^{q}$.
Following \cite{HydonMansfield2004} we have that  if the operator $\vec{D}_{P}$ 
is self-adjoint with respect to the $\ell^{2}\left( \Z \right)$ norm
then the difference equation of order $r$ defined by
\begin{equation}
    P\left[ x_{n} \right]=0, \quad P\in\mathcal{A}^{r}
    \label{eq:Peq}
\end{equation}
is Lagrangian.
To prove that the Fr\'echet derivative defined by a given
difference equation $P\in\mathcal{A}^{r}$
is or not self-adjoint is an algorithmic task which can be always accomplished
with a finite number of steps.
However, a negative answer to this procedure is not definitive.
Indeed, even though the operator $\vec{D}_{P}$ is not self-adjoint
a difference function $\mu\in\mathcal{A}^{r-1}$ such that 
$\vec{D}_{\mu \cdot P}$ is self-adjoint might exist.
Now since $\mu\left[ x_{n} \right]\in\mathcal{A}^{r-1}$ we have
that the difference equation defined by:
\begin{equation}
    \mu\left[ x_{n} \right] P\left[ x_{n} \right]=0,
    \label{eq:multiplier}
\end{equation}
is equivalent to equation \eqref{eq:Peq}.
The difference function $\mu\left[ x_{n} \right]$ is called a 
\emph{multiplier}.

Therefore, \emph{to disprove the existence of a Lagrangian for the difference
equation \eqref{eq:Peq} we must be able to prove that
the operator $\vec{D}_{\mu\cdot P}$ is not self-adjoint \emph{for every}
multiplier $\mu\left[ x_{n} \right]$}.

On the other hand, we notice that with our method, multipliers are 
unessential.
Indeed, the annihilation operators \eqref{eq:annihil} and \eqref{eq:annihil2}
are uniquely defined by the solutions of the difference equations
\eqref{eq:Peq} with respect to the highest or the lowest shift of $x_{n}$,
namely equations \eqref{eq:recgen} and \eqref{eq:recgeninv}.
Therefore, the conditions stated in Theorem \ref{thm:elexist} are 
independent of the value of any possible multiplier $\mu\in\mathcal{A}^{r-1}$,
since equations \eqref{eq:recgen} and \eqref{eq:recgeninv}
are unchanged upon multiplication.
So, the outcome of the method presented in Section \ref{sec:method}
is definitive up to point transformations in the variable $x_{n}$.

Moreover, in an upcoming paper we will consider the problem of 
classification of variational difference equations of a given order
with some mild additional
assumption on the functional form of the difference equation 
(e.g. reversibility).
In that contest the conditions of Theorem \ref{thm:elexist} will
act as constraint on the functional form of the function $F$ in
\eqref{eq:recgen}.
We remark that this problem is non-trivial since using the 
systems \eqref{eq:lagrcondsys} and \eqref{eq:lagrcondsysinv} as a 
classification tool requires to derive and solve systems of
nonlinear partial difference equations.

Finally, work is in progress to extend the method of the 
annihilation operators to more general discrete variational problems, 
like systems of difference equations.

\section*{Acknowledgements}

We thank Prof. Nalini Joshi, Prof. Maria Clara Nucci and Dr. Dinh Tran 
for the interesting and fruitful discussion during the preparation of this paper.
Moreover, we thank Dr. Dinh T. Tran for pointing out the examples 
contained in Subsection \ref{sub:tranex}.
Finally, we express our gratitude to the anonymous referees whose comments 
and suggestions led to a great improvement of the paper.

GG is supported by the Australian Research Council through 
Nalini Joshi's Australian Laureate Fellowship grant FL120100094.

\appendix

\section{Conditions on the existence of a Lagrangian for equations \eqref{eq:ELTran}}
\label{app:lemmaproof}

To end the proof of Lemma \ref{lem:lagrnihil} we must show that the Lagrangians
\eqref{eq:LTran} arise from the conditions \eqref{eq:lagrcondsys}.
As the computations are analogous to those presented in Subsection
\ref{sub:tranex} we will limit ourselves to present the form  of the 
forward annihilation operators \eqref{eq:annihil} and of the corresponding
conditions on the existence of a Lagrangian \eqref{eq:lagrcondsys}.
The interested reader can check using the ansatz \eqref{eq:Lkdvsimp} that
these conditions yield the desired Lagrangians.

In Subsection \ref{sub:tranex} we already discussed the Lagrangian
\eqref{eq:Lkdv}, so we start from the Lagrangian \eqref{eq:Lpluskdv}
and its Euler--Lagrange equations \eqref{eq:ELpluskdv}.
In the case of equation \eqref{eq:ELpluskdv} the forward annihilation 
operators \eqref{eq:annihil} have the following expression:
\begin{subequations}
    \begin{align}
        \vec{A}_{m}^{+} &= \frac{1}{D^{2}_{\text{plusKdV}}}
        \pdv{}{x_{n-m}},
        \quad m\neq p,q,
        \label{eq:Amgenpluskdv}
        \\
        \vec{A}_{p}^{+} &= \frac{1}{D^{2}_{\text{plusKdV}}}
        \left[\pdv{}{x_{n-p}}- \left( x_{n}+x_{n-p-q} \right)^{2}\pdv{}{x_{n-p-q}}  \right],
        \label{eq:Amppluskdv}
        \\
        \vec{A}_{q}^{+} &= \frac{1}{D^{2}_{\text{plusKdV}}} 
        \left[\pdv{}{x_{n-p}}-\left( x_{n}+x_{n-p-q} \right)^{2}\pdv{}{x_{n-p-q}}  \right],
        \label{eq:Amqpluskdv}
    \end{align}
    \label{eq:annihiloppluskdv}
\end{subequations}
where:
\begin{equation}
    D_{\text{plusKdV}} = 1-\left( x_{n}+x_{n-p-q} \right)\left(x_{n+p}+x_{n-p}+x_{n+q}+x_{n-q}\right).
    \label{eq:Ddefplus}
\end{equation}
Using these operators we have that the differential
conditions on the Lagrangian \eqref{eq:lagrcondsys} are given by:
\begin{subequations}
    \begin{gather}
        \pdv{\ell}{x_{n-m},x_{n-p-q}} =0, \quad m\neq p,q,
        \label{eq:annpluskdvmgen}
        \\
        \begin{aligned}
        \left( x_{{n}}+x_{{n-q-p}} \right) ^{2}{\frac {\partial ^{2} \ell}{\partial x_{{n-q-p}}^{2}}}
        &-{\frac {\partial ^{2} \ell}{\partial x_{{n-q-p}}\partial x_{{n-p}}}}
        \\ 
        &+ 2\left( x_{{n-q-p}}+x_{{n}} \right) {\frac {\partial \ell}{\partial x_{{n-q-p}}}} =0,
        \end{aligned}
        \quad m=p,
        \label{eq:annpluskdvp}
        \\
        \begin{aligned}
        \left( x_{{n}}+x_{{n-q-p}} \right) ^{2}{\frac {\partial ^{2} \ell}{\partial x_{{n-q-p}}^{2}}}
        &-{\frac {\partial ^{2} \ell}{\partial x_{{n-q-p}}\partial x_{{n-q}}}}
        \\
        &+ 2\left( x_{{n-q-p}}+x_{{n}} \right) {\frac {\partial \ell}{\partial x_{{n-q-p}}}} =0,
        \end{aligned}
        \quad m=q.
        \label{eq:annpluskdvq}
    \end{gather}
    \label{eq:annpluskdv}
\end{subequations}
The Lagrangian obtained in \cite{TranvanderKampQusipel2016} arises solving
the system \eqref{eq:annpluskdv}.

In the case of equation \eqref{eq:ELlv} the forward annihilation 
operators \eqref{eq:annihil} have the following expression:
\begin{subequations}
    \begin{align}
        \vec{A}_{m}^{+} &= D_{\text{LV}}
        \pdv{}{x_{n-m}},
        \quad m\neq p,q,
        \label{eq:AmgenLV}
        \\
        \vec{A}_{p}^{+} &= D_{\text{LV}}
        \left[\pdv{}{x_{n-p}}+ \left(1+e^{x_{n}-x_{n-p-q}}\right)\pdv{}{x_{n-p-q}}  \right],
        \label{eq:AmpLV}
        \\
        \vec{A}_{q}^{+} &= D_{\text{LV}}
        \left[\pdv{}{x_{n-p}}-\left( 1+e^{x_{n}-x_{n-p-q}} \right)\pdv{}{x_{n-p-q}}  \right],
        \label{eq:AmqLV}
    \end{align}
    \label{eq:annihilopLV}
\end{subequations}
where:
\begin{equation}
    D_{\text{LV}} = \frac{e^{x_{{n+q}}+x_{{n-q}}-x_{{n+p}}-x_{{n-p}}+x_{{n-p-q}}-x_{{n}}}}{e^{x_{{n+q}}+x_{{n-q}}-x_{{n+p}}-x_{{n-p}}}\left( 1+e^{{+x_{{n-p-q}}-x_{{n}}}} \right)-1}.    
    \label{eq:Ddeflv}
\end{equation}
Using these operators we have that the differential
conditions on the Lagrangian \eqref{eq:lagrcondsys} are given by:
\begin{subequations}
    \begin{gather}
        \pdv{\ell}{x_{n-m},x_{n-p-q}} =0, \quad m\neq p,q,
        \label{eq:annLVmgen}
        \\
        \begin{aligned}
            \left(1+e^{x_{n}-x_{n-p-q}}\right){\frac {\partial ^{2} \ell}{\partial x_{{n-q-p}}^{2}}}
        &+{\frac {\partial ^{2} \ell}{\partial x_{{n-q-p}}\partial x_{{n-p}}}}
        \\ 
        &- e^{x_{n}-x_{n-p-q}} {\frac {\partial \ell}{\partial x_{{n-q-p}}}} =0,
        \end{aligned}
        \quad m=p,
        \label{eq:annLVp}
        \\
        \begin{aligned}
            \left(1+e^{x_{n}-x_{n-p-q}}\right){\frac {\partial ^{2} \ell}{\partial x_{{n-q-p}}^{2}}}
        &-{\frac {\partial ^{2} \ell}{\partial x_{{n-q-p}}\partial x_{{n-q}}}}
        \\
        &- e^{x_{n}-x_{n-p-q}} {\frac {\partial \ell}{\partial x_{{n-q-p}}}} =0,
        \end{aligned}
        \quad m=q.
        \label{eq:annLVq}
    \end{gather}
    \label{eq:annLV}
\end{subequations}
The Lagrangian obtained in \cite{TranvanderKampQusipel2016} arises solving
the system \eqref{eq:annLV}.

In the case of equation \eqref{eq:ELliouville} the forward annihilation 
operators \eqref{eq:annihil} have the following expression:
\begin{subequations}
    \begin{align}
        \vec{A}_{m}^{+} &= D_{\text{ASdL}}
        \pdv{}{x_{n-m}},
        \quad m\neq p,q,
        \label{eq:Amgenliouville}
        \\
        \vec{A}_{p}^{+} &= D_{\text{ASdL}}
        \left[\pdv{}{x_{n-p}}- \left(1+e^{x_{n}+x_{n-p-q}}\right)\pdv{}{x_{n-p-q}}  \right],
        \label{eq:Ampliouville}
        \\
        \vec{A}_{q}^{+} &= D_{\text{ASdL}}
        \left[\pdv{}{x_{n-p}}-\left( 1+e^{x_{n}+x_{n-p-q}} \right)\pdv{}{x_{n-p-q}}  \right],
        \label{eq:Amqliouville}
    \end{align}
    \label{eq:annihilopliouville}
\end{subequations}
where:
\begin{equation}
    D_{\text{ASdL}} = {\frac {e^{x_{n-p-q}+x_{n}-x_{n+p}-x_{n+q}-x_{n-q}-x_{n-p}}}{%
                \left( 1+{{  e}^{x_{{n-p-q}}+x_{{n}}}} \right)  
                \left( 1+{{  e}^{x_{{n-p-q}}+x_{{n}}}}-{{  e}^{-x_{{n+p}}-x_{{n+q}}-x_{{n-q}}-x_{{n-p}}}} \right) }}
    \label{eq:Ddefliouville}
\end{equation}
Using these operators we have that the differential
conditions on the Lagrangian \eqref{eq:lagrcondsys} are given by:
\begin{subequations}
    \begin{gather}
        \pdv{\ell}{x_{n-m},x_{n-p-q}} =0, \quad m\neq p,q,
        \label{eq:annliouvillemgen}
        \\
        \begin{aligned}
            \left(1+e^{x_{n}+x_{n-p-q}}\right){\frac {\partial ^{2} \ell}{\partial x_{{n-q-p}}^{2}}}
        &-e^{x_{n}+x_{n-p-q}}{\frac {\partial ^{2} \ell}{\partial x_{{n-q-p}}\partial x_{{n-p}}}}
        \\ 
        &-  {\frac {\partial \ell}{\partial x_{{n-q-p}}}} =0,
        \end{aligned}
        \quad m=p,
        \label{eq:annliouvillep}
        \\
        \begin{aligned}
            \left(1+e^{x_{n}+x_{n-p-q}}\right){\frac {\partial ^{2} \ell}{\partial x_{{n-q-p}}^{2}}}
            &-e^{x_{n}+x_{n-p-q}}{\frac {\partial ^{2} \ell}{\partial x_{{n-q-p}}\partial x_{{n-q}}}}
            \\
            &-{\frac {\partial \ell}{\partial x_{{n-q-p}}}} =0,
        \end{aligned}
        \quad m=q.
        \label{eq:annliouvilleq}
    \end{gather}
    \label{eq:annliouville}
\end{subequations}
The Lagrangian obtained in \cite{TranvanderKampQusipel2016} arises solving
the system \eqref{eq:annliouville}.
This ends the proof of Lemma \ref{lem:lagrnihil}.

\bibliographystyle{plain}
\bibliography{bibliography}

\end{document}